\documentclass[twoside,11pt]{article}
\usepackage{amsmath,amsthm,amssymb,bm,epic}
\usepackage[reset,a4paper,vmargin=3truecm,hmargin=2truecm]{geometry}
\usepackage{amsmath,amsthm,amssymb,amscd,epic}

\usepackage[all]{xy}
\usepackage{amssymb, amsmath,bm}
\numberwithin{equation}{section}

\theoremstyle{plain}
\newtheorem{thm}{Theorem}[section]
\newtheorem{lem}[thm]{Lemma}
\newtheorem{prop}[thm]{Proposition}
\newtheorem{cor}[thm]{Corollary}
\newtheorem{conjecture}[thm]{Conjecture}

\theoremstyle{definition}

\newtheorem{prob}{Problem}[section]

\theoremstyle{remark}
\newtheorem{rem}{Remark}[section]

\newcommand{\Z}{\mathbb{Z}} % integers
\newcommand{\R}{\mathbb{R}} % real numbers
\newcommand{\C}{\mathbb{C}} % complex numbers

\newcommand{\kakko}[1]{\left(#1\right)}

\newcommand{\set}[3][big]% set
{\csname#1l\endcsname\{#2\,\csname#1\endcsname|\,#3\csname#1r\endcsname\}}
\newcommand{\Set}[2]{
\setbox1=\hbox{$\displaystyle\left\{#1\right\}$}
\setbox2=\hbox{$\displaystyle\left\{#2\right\}$}
\left\{
#1
\ \vrule width .5pt height \ifdim\ht1>\ht2 \ht1 \else \ht2 \fi\
#2
\right\}}

\DeclareMathOperator{\rank}{rank}

\title{Symmetry of asymmetric quantum Rabi models}

\author{Masato Wakayama
}

\begin{document}

\maketitle

\abstract
The aim of this paper is a better understanding for the eigenstates of the asymmetric quantum Rabi model by Lie algebra representations of $\mathfrak{sl}_2$. We define a second order element of the universal enveloping algebra $\mathcal{U}(\mathfrak{sl}_2)$ of $\mathfrak{sl}_2(\R)$, which,  through the action of  a certain infinite dimensional representation of $\mathfrak{sl}_2(\mathbb{R})$, provides a picture of the asymmetric quantum Rabi model equivalent to the one drawn by confluent Heun ordinary differential equations. Using this description, we prove the existence of level crossings in the spectral graph of the asymmetric quantum Rabi model when the symmetry-breaking parameter $\epsilon$ is equal to $\frac12$, and conjecture a formula  that ensures likewise the presence of level crossings for  general $\epsilon \in \frac12\Z$. This result on level crossings was demonstrated numerically by Li and Batchelor in 2015, investigating an earlier empirical observation by Braak (2011). The first analysis of
the degenerate spectrum was given for the symmetric quantum Rabi model by Ku\'s in 1985. In our picture, we find a certain reciprocity (or $\Z_2$-symmetry) for $\epsilon \in \frac12\Z$ if the spectrum is described by representations of $\frak{sl}_2$. We further discuss briefly the non-degenerate part of the exceptional  spectrum from the viewpoint of infinite dimensional representations  of $\mathfrak{sl}_2(\mathbb{R})$ having lowest weight vectors.  

\;
\noindent
{\bf 2010 Mathematics Subject Classification:} 
{\it Primary} 34L40, {\it Secondary} 81Q10, 34M05, 81S05.

\;
\noindent
{\bf Keywords and phrases:}
{quantum Rabi models, level crossings, confluent Heun differential equations, exceptional spectrum, irreducible representations.}

\,

%=========================================================================
\section{Introduction}
%=========================================================================

The quantum Rabi model (QRM) is the fully quantized version \cite{JC1963} of the original Rabi model (for a review see \cite{bcbs2016}) and is known to be the simplest model used in quantum optics to describe the interaction of light and matter. As such, it appears ubiquitously in various quantum systems including cavity and circuit quantum electrodynamics, quantum dots etc. (see e.g \cite{HR2008}).  The Hamiltonian of the quantum Rabi model ($\hbar=1$) reads 
\begin{equation}
H_{\rm{Rabi}}=\omega a^\dag a+\Delta \sigma_z +g\sigma_x(a^\dag+a),
\label{QRM}
\end{equation}
with $a^\dag = (x-\partial_x)/\sqrt2, a= (x+\partial_x)/\sqrt2\, (\partial_x:=\frac{d}{dx})$
and  $\sigma_x = \begin{bmatrix}
 0 & 1  \\
 1 &  0 
\end{bmatrix}, 
\sigma_z= \begin{bmatrix}
 1 & 0  \\
 0 &  -1
\end{bmatrix}$ 
are the Pauli matrices, $2\Delta$ is the energy difference between the
two levels, and $g$ denotes the coupling strength between the two-level system and the bosonic mode with frequency $\omega$. It is well-known that the Jaynes-Cummings model \cite{JC1963} (see \cite{HR2008} for a comprehensive overview and recent applications), the RWA (rotating-wave approximation) of the quantum Rabi model,  has a $U(1)$-symmetry and is known to be integrable. Although the quantum Rabi model has no such $U(1)$-symmetry, it has a $\Z_2$-symmetry (parity). Using this $\Z_2$-symmetry, D. Braak \cite{B2011PRL} has shown the integrability of the quantum Rabi model in 2011.  
In the present paper, we study the spectrum of the following asymmetric quantum Rabi model \cite{Le2016} (called ``generalized"  quantum Rabi model in \cite{B2011PRL,LB2015JPA}, ``biased'' in \cite{BLZ2015JPALetter} and  ``driven" in \cite{La2013}) with broken $\Z_2$-symmetry. This asymmetric model provides actually a more realistic description of circuit QED experiments employing flux qubits than the QRM itself \cite{Ni2010}: 
\begin{equation}
H_{\rm{Rabi}}^\epsilon=\omega a^\dag a+\Delta \sigma_z +g\sigma_x(a^\dag+a) + \epsilon \sigma_x.
\label{aH}
\end{equation}
We set $\omega=1$ in the following.
There are several studies on the spectrum of the asymmetric quantum Rabi model \cite{B2013AP,B2013MfI, BLZ2015JPALetter, LB2015JPA, LB2016JPA}. Here we are particularly interested in the level crossings visible in the spectral graph, i.e. the set of discrete eigenvalues of $H_{\rm{Rabi}}^\epsilon$ as function of the coupling strength $g$. The presence of level crossings in the asymmetric model is highly non-trivial because 
the additional term $\epsilon \sigma_x$ breaks the $\Z_2$-symmetry which couples the bosonic mode and the two-level system by allowing spontaneous tunneling between the two atomic states. Although the $\Z_2$-symmetry may be recovered by ``integrable embedding'' \cite{B2011PRL}, so that the same solution method as used for the ``symmetric" QRM can be applied to the asymmetric model as well  \cite{B2013AP}, the broken $\Z_2$-symmetry removes all level crossings from the spectral graph for general $\epsilon$, rendering the model non-integrable according to the criterion for quantum integrability proposed in \cite{B2011PRL}. Without the symmetry, there seem to be no invariant subspaces anymore whose respective spectral graphs may intersect to create ``accidental'' degeneracies in the spectrum for specific values of the coupling. 
 Surprisingly, the crossings have been numerically confirmed in \cite{LB2015JPA} when the parameter $\epsilon$ is a half-integer, hinting at a hidden symmetry present in this case. 

In general, these level crossings(or degeneracies of eigenstates) have been studied also in pure mathematics for different reasons \cite{H2009IUMJ, HH2012, HS2013B}. Especially, the ``gamma factor" (a terminology used in number theory) of the spectral zeta function \cite{Sugi2016} for the quantum Rabi model  may correspond to the so-called exceptional spectrum of the model \cite{B2011PRL, W2014MfI}. It would be very interesting if this gamma factor (exceptional eigenstates) could be interpreted as, e.g. a topological property of the states as in the case of the Selberg zeta functions for discontinuous groups (see, e.g \cite{Hej1976}). Also, we expect that representation theoretic treatment of the lever crossings developed here may help the future study of monodoromy problems for the confluent Heun ODE.

In the following we shall employ the theory worked out previously 
\cite{WY2014JPA}. Using this description, we prove the presence of level crossings of the asymmetric quantum Rabi model for $\epsilon= \frac12$.  This result has been demonstrated numerically by Li and Batchelor \cite{LB2015JPA}. 
As a subsequent observation, we give a representation theoretic description of such crossings for general $\epsilon \in \frac12\Z$ by means of finite dimensional representations of $\mathfrak{sl}_2$. Thus we may observe a certain reciprocity (or a $\Z_2$-symmetry) behind this structure. We leave the proof of the crossings for general $\epsilon \in \frac12\Z$ and a detailed mathematical investigation for future study but discuss briefly the non-degenerate part of the exceptional  spectrum of the model from the  viewpoint of  infinite dimensional irreducible submodules (or subquotients) of the non-unitary principal series such as the (holomorphic) discrete series representations of $\mathfrak{sl}_2(\mathbb{R})$.  

%=========================================================================
\section{Confluent Heun's picture of the model}\label{CHP}
%=========================================================================

In this section, we recall the confluent Heun picture of the asymmetric quantum Rabi model. 
We may assume $\omega=1$ without loss of generality in the Hamiltonian $H_{\rm{Rabi}}^\epsilon =H_{\rm{Rabi}}^\epsilon(g, \Delta, \omega)$ of the asymmetric quantum Rabi model. 
In what follows, we will use the Bargmann representation of the boson operators \cite{Bar1961}. Here $a^\dag$ and $a$ 
are realized as the multiplication and differentiation operators over the complex variable:   
$a^\dag = (x-\partial_x)/\sqrt2 \to z$ and $a= (x+\partial_x)/\sqrt2 \to \partial_z:=\frac{d}{dz}$. These operators act on the Hilbert 
space $\mathcal{B}$ of entire functions equipped with the inner product 
$$
(f|g)= \frac1\pi \int_{\C} \overline{f(z)}g(z)e^{-|z|^2}d({\rm{Re}}(z))d({\rm{Im}}(z)).
$$
In this Bargmann picture, the Hamiltonian 
$$H_{\rm{Rabi}}^\epsilon \rightarrow
\tilde{H}_{\rm{Rabi}}^\epsilon = z\frac{d}{dz}I+ g\big(z+\frac{d}{dz}\big)\sigma_x+\Delta \sigma_z+\epsilon\sigma_x.
$$ 
Therefore, by the standard procedure (see e.g. \cite{B2013AP, LB2015JPA}), we observe that the Schr\"odinger equation 
$H_{\rm{Rabi}}^\epsilon\varphi=\lambda \varphi$  is equivalent to the system of first order differential equations 
\begin{equation*}
\tilde{H}_{\rm{Rabi}}^\epsilon\psi=\lambda \psi, \quad 
\psi=  \begin{bmatrix}
 \psi_{1}(z) \\
 \psi_{2}(z)
\end{bmatrix}.
\end{equation*}
Hence, in order to have an eigenstate of $H_{\rm{Rabi}}^\epsilon$, it is sufficient to obtain an eigenstate
$\psi \in \mathcal{B}$, that is, $\mathbf{BI}$: $(\psi_i|\psi_i) <\infty $, and $\mathbf{BII}$: $\psi_i$ are holomorphic everywhere in the whole complex plane $\C$ for $i=1,2$.

Let $W_\lambda$ be the eigenspace of $H_{\rm{Rabi}}^\epsilon$, whence of $\tilde{H}_{\rm{Rabi}}^\epsilon$, corresponding to the  eigenvalue $\lambda$. It is obvious that $\dim W_\lambda \leq 2$. 
Upon writing $f_\pm:= \psi_1\pm \psi_2$, we have 
\begin{equation}
\begin{cases}\label{SDE1}
(z+g)\frac{d}{dz} f_++ (gz+\epsilon-\lambda)f_+ +\Delta f_- =0,\\
(z-g)\frac{d}{dz}f_- - (gz+\epsilon+\lambda)f_- +\Delta f_+=0.
\end{cases}
\end{equation}
It is elementary to see that this system of first order equations is regular singular at $z=\pm g$, irregular singular at $\infty$, and no other singularities. However, since the irregular singular point at $\infty$  has rank $1$ \cite{Ince1956}, we observe that  the asymptotic expansions of $f_\pm\,(i=1,2)$ for $z\to \infty$ reads $f_{\pm}(z)= e^{cz} z^\rho(c_0+ c_1/z+\cdots)$ with some constants $c_\pm, \rho_\pm, c_{\pm,1}, c_{\pm,2},\ldots$, whence  the condition $\mathbf{BI}$ is always satisfied, because all functions growing like $e^{cz}$ for $|z| \to \infty$ satisfy $\mathbf{BI}$. This fact implies that  the analyticity of $\psi$, resp. $f_\pm$ at the regular singular points $z=\pm g$, corresponding to condition $\mathbf{BII}$, is sufficient for $f_\pm\in {\mathcal B}$. We leave the detailed discussion to \cite{B2013AP}. 

\;

Now, from these equations, we can get two sets of solutions for $f_+(z)$ and $f_-(z)$ as follows:

\noindent
$\bullet$ Substitute $\phi_{1,\pm}(z):=e^{gz}f_\pm(z)$ into the equations and eliminate $\phi_{1,-}(z)$ from these equations. Further, put $\phi_1(x)=\phi_{1,+}(z)$, where $x:=(g+z)/2g$. Then we have 
$\mathcal{H}_1^\epsilon(\lambda) \phi_1(x)=0$, where
\begin{align*}%\label{Heun+}
\mathcal{H}_1^\epsilon(\lambda):=
 \frac{d^2}{dx^2} + & \Big\{-4g^2+\frac{1-(\lambda+g^2)+\epsilon}{x}+\frac{1-(\lambda+g^2+1)-\epsilon}{x-1}\Big\}\frac{d}{dx}\\
& +\frac{4g^2(\lambda+g^2-\epsilon)x+\mu+4\epsilon g^2-\epsilon^2}{x(x-1)}.
\end{align*}
\smallskip

\noindent
$\bullet$ Substitute $\phi_{2,\pm}(z):=e^{-gz}f_\pm(z)$ into the equations and eliminate $\phi_{2,-}(z)$ from these equations. Further, put $\phi_2(x)=\phi_{2,+}(z)$, where $x:=(g-z)/2g$. Then we have 
$\mathcal{H}_2^\epsilon(\lambda) \phi_2(x)=0$, where
\begin{align*}%\label{Heun-}
\mathcal{H}_2^\epsilon(\lambda):=
\frac{d^2}{dx^2}+ & \Big\{-4g^2+\frac{1-(\lambda+g^2+1)-\epsilon}{x}+\frac{1-(\lambda+g^2)+\epsilon}{x-1}\Big\}\frac{d}{dx}\\
& +\frac{4g^2(\lambda+g^2-1+\epsilon)x+\mu-4\epsilon g^2-\epsilon^2}{x(x-1)}.
\end{align*}

\medskip
\noindent
Here  the constant $\mu$ is defined by 
\begin{equation} \label{Mu}
\mu := (\lambda+g^2)^2-4g^2(\lambda+g^2)-\Delta^2. 
\end{equation}
and contributes the accessory parameter of each equation $\mathcal{H}_j^\epsilon(\lambda) \phi_j(x)=0\;(j=1,2)$ (e.g. \cite{SL2000}).
Notice that these $\mathcal{H}_j^\epsilon(\lambda)\, (j=1,2)$ are the confluent Heun differential operators and give regular singularities at $x=0, 1$ and irregular one at $x=\infty$ to the corresponding equations (the Heun picture of asymmetric quantum Rabi models). 
Since the procedure of derivation of these equations for $\phi_j(x)\, (j=1,2)$ respectively is standard, we leave its details to the reader. We close the section by observing the exponents of these equations. 
\begin{lem}\label{Exponents}
The exponents $\rho$ of each regular singular points for the equations 
$\mathcal{H}_j^\epsilon(\lambda) \phi_j(x)=0\,(j=1,2)$ are respectively given as follows: 
\begin{align*}
\mathcal{H}_1^\epsilon(\lambda)\,:\; & \rho= 0, \lambda+g^2-\epsilon \; (x=0), \quad \rho= 0, \lambda+g^2+1+\epsilon \; (x=1),\\
\mathcal{H}_2^\epsilon(\lambda)\,:\; & \rho= 0, \lambda+g^2+1+\epsilon \; (x=0), \quad \rho=  0, \lambda+g^2-\epsilon \; (x=1).
\end{align*}
In particular, both exponents at $x=0,1$ of $\mathcal{H}_1^\epsilon(\lambda)$ (resp. $\mathcal{H}_2^\epsilon(\lambda)$) are integers iff one of two cases, i.e. either $\lambda+g^2, \epsilon \in \Z$ or 
$\lambda+g^2, \epsilon \in \Z+\frac12$ holds. \qed 
\end{lem}

%=========================================================================
\section{Representations of $\mathfrak{sl}_2$}
%=========================================================================

We recall briefly the representation theoretic setting for the Lie algebra $\mathfrak{sl}_2$ in order to elucidate the  symmetry behind the asymmetric quantum Rabi model. 

Let $H, E$ and $F$ be the standard generators of $\mathfrak{sl}_2$ defined by 
\begin{align*}
H= \begin{bmatrix}
 1 & 0  \\
 0 &  -1 
\end{bmatrix},\quad
E= \begin{bmatrix}
 0 & 1  \\
 0 &  0 
\end{bmatrix},\quad
F= \begin{bmatrix}
 0 & 0  \\
 1 &  0 
\end{bmatrix}. 
\end{align*}
These  satisfy the commutation relations
$$
[H,\, E] = 2E, \,\, [H,\, F] = -2F, \,\,  [E,\, F] = H. 
$$
Let  $a\in \C$.  Put $\mathbf{V}_{1}:=x^{-\frac14} \mathbb{C}[x, x^{-1}]$
and $\mathbf{V}_{2}:=x^{\frac14} \mathbb{C}[x, x^{-1}]$. 
Consider the following algebraic actions of $\mathfrak{sl}_2$ on  $\mathbf{V}_{j},\, (j=1,2)$ defined by 
\begin{align*}
\varpi_a(H) :=2x\partial_x+\frac12,\quad 
\varpi_a(E) :=x^2\partial_x+\frac12(a+\frac12)x,\quad
\varpi_a(F) := -\partial_x+\frac12(a-\frac12)x^{-1}. 
\end{align*}
These operators indeed act on the spaces  $\mathbf{V}_j,\, (j=1,2)$
and define infinite dimensional representations (non-unitary principal series representations) of $\mathfrak{sl}_2$. Write  $\varpi_{j,a}=\varpi_a|_{\mathbf{V}_{j}}$ and put $e_{1,n}:=x^{n-\frac14}, \, e_{2,n}:=x^{n+\frac14}$. Then we have

\smallskip
\noindent
$\bullet$ {\it the spherical principal series}: on $\mathbf{V}_{1,a}:=\mathbf{V}_{1}=\oplus_{n\in \Z}\,\C\cdot e_{1,n} $ 
%\begin{eqnarray*}
\begin{equation*}
\begin{cases}
\varpi_{1,a}(H)e_{1,n} = 2ne_{1,n},\\
\varpi_{1,a}(E)e_{1,n} = \big(n+\frac{a}{2}\big)e_{1,n+1},\\
\varpi_{1,a}(F)e_{1,n} = \big(-n+\frac{a}{2}\big)e_{1,n-1}.
\end{cases}
%\text{on} \; \mathbf{V}_{1,a}=\oplus_{n\in \mathbb{Z}}\mathbb{C}\cdot e_{1,n}  
%\text{(the spherical principal series)}
%\end{eqnarray*}
\end{equation*}
and 

%\begin{eqnarray*}
\smallskip
\noindent
$\bullet$ {\it the non-spherical principal series}: on $\mathbf{V}_{2,a}:=\mathbf{V}_{2}=\oplus_{n\in \Z} \, \C \cdot e_{2,n}$
\begin{equation*}
\begin{cases} 
\varpi_{2,a}(H)e_{2,n} &= (2n+1)e_{2,n},\\
\varpi_{2,a}(E)e_{2,n} &= \big(n+\frac{a+1}{2}\big)e_{2,n+1},\\
\varpi_{2,a}(F)e_{2,n} &= \big(-n+\frac{a-1}{2}\big)e_{2,n-1}.
\end{cases}
\end{equation*}

Note that $(\varpi_{1,a},\mathbf{V}_{1})$ (resp. $(\varpi_{2,a},\mathbf{V}_{2}$) is irreducible when $a\not\in 2\mathbb{Z}$ (resp. $a\not\in 2\mathbb{Z}-1$) and there is an equivalence between $\varpi_{j,a}$ and $\varpi{j,2-a}$ under the same condition. 
 
For a non-negative integer $m$, define subspaces 
$\mathbf{D}^{\pm}_{2m},\mathbf{F}_{2m-1}$ of $\mathbf{V}_{1,2m}(=\mathbf{V}_{1})$, 
and  $\mathbf{D}^{\pm}_{2m+1},\mathbf{F}_{2m}$ of $\mathbf{V}_{2,2m+1}(=\mathbf{V}_{2})$ respectively by 
\begin{equation*}
\mathbf{D}^{\pm}_{2m}:=\bigoplus_{n\geq  m}\mathbb{C}\cdot e_{1,\pm n},\:\:\:
\mathbf{F}_{2m-1}:=\bigoplus_{-m+1\leq n\leq m-1}\mathbb{C}\cdot e_{1,n},
\end{equation*}
\begin{equation*}
\mathbf{D}^{-}_{2m+1}:=\bigoplus_{n\geq m+1}\mathbb{C}\cdot e_{2,-n},\:\:\:
\mathbf{D}^{+}_{2m+1}:=\bigoplus_{n\geq m}\mathbb{C}\cdot e_{2,n},\:\:\:
\mathbf{F}_{2m}:=\bigoplus_{-m\leq n\leq m-1}\mathbb{C}\cdot e_{2,n}.
\end{equation*}
The spaces $\mathbf{D}^{\pm}_{2m}$ (resp. $\mathbf{D}^{\pm}_{2m+1}$) are invariant  under the action $\varpi_{1,2m}(X)$ (resp. $\varpi_{1,2m+1}(X))$ $(X\in \mathfrak{sl}_2)$, and define irreducible representations known to be equivalent to (holomorphic and anti-holomorphic) discrete series for $m>0$ of $\mathfrak{sl}_2(\mathbb{R})$  (see e.g. \cite{L, HT1992}). Moreover,  the finite dimensional space $\mathbf{F}_{m}$ ($\dim_{\mathbb{C}} \mathbf{F}_{m}=m$),  is invariant and defines irreducible representation of $\mathfrak{sl}_2$ for $a=2-2m$ when $j=1$ and $a=1-2m$ when $j=2$, respectively. 

\bigskip
\medskip

\begin{figure}[http]
\quad 
$\mathbf{V}_{1,2m\, (m\in\Z)}: \;$
\begin{picture}(50,10)
\dottedline{4}(0,0)(70,0)
\dottedline{4}(120,0)(170,0)
\dottedline{4}(190,0)(240,0)
\dottedline{4}(270,0)(340,0)
\thicklines
\jput(90,0){\makebox(0,0){$\bullet$}}
\jput(120,0){\makebox(0,0){$\bullet$}}
\jput(180,0){\makebox(0,0){$\bullet$}}
\jput(240,0){\makebox(0,0){$\bullet$}}
\jput(270,0){\makebox(0,0){$\bullet$}}
\drawline(10,15)(80,15)(90,0)
\drawline(120,0)(130,15)(230,15)(240,0)
\drawline(270,0)(280,15)(350,15)
\put(20,20){$\mathbf{D}_{2m}^{-}$}
\put(170,20){$\mathbf{F}_{2m-1}$}
\put(310,20){$\mathbf{D}_{2m}^{+}$}
\put(70,-15){$-2m$}
\put(115,-15){$-2m+2$}
\put(180,-15){$0$}
\put(210,-15){$2m-2$}
\put(270,-15){$2m$}
\end{picture}
\end{figure}

\bigskip
\bigskip

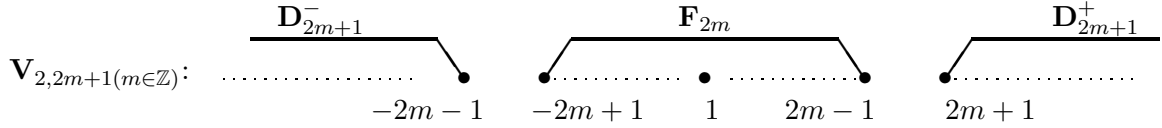
\begin{figure}[http]
\quad 
$\mathbf{V}_{2,2m+1(m\in\Z)}$: \;
\begin{picture}(50,10)
\dottedline{4}(0,0)(70,0)
\dottedline{4}(120,0)(170,0)
\dottedline{4}(190,0)(240,0)
\dottedline{4}(270,0)(340,0)
\thicklines
\jput(90,0){\makebox(0,0){$\bullet$}}
\jput(120,0){\makebox(0,0){$\bullet$}}
\jput(180,0){\makebox(0,0){$\bullet$}}
\jput(240,0){\makebox(0,0){$\bullet$}}
\jput(270,0){\makebox(0,0){$\bullet$}}
\drawline(10,15)(80,15)(90,0)
\drawline(120,0)(130,15)(230,15)(240,0)
\drawline(270,0)(280,15)(350,15)
\put(20,20){$\mathbf{D}_{2m+1}^{-}$}
\put(170,20){$\mathbf{F}_{2m}$}
\put(310,20){$\mathbf{D}_{2m+1}^{+}$}
\put(55,-15){$-2m-1$}
\put(115,-15){$-2m+1$}
\put(180,-15){$1$}
\put(210,-15){$2m-1$}
\put(270,-15){$2m+1$}
\end{picture}
\bigskip
\caption{Weight decompositions of $\mathbf{V}_{j,m}$ for $m\in \mathbb{Z}$}
\end{figure}

For $a\in \mathbb{Z}$, summarizing the actions $\varpi_{j, a}$ of $\mathfrak{sl}_2$, we recall  the following irreducible decomposition of $(\varpi_a, \, \mathbf{V}_{j,a})\, (a=m\in \Z)$. 
%%%%%%%%%%%%%%%%%%%
\begin{lem} \label{Reducible}
Let $m\in \mathbb{Z}_{\geq 0}$.  
\begin{enumerate}
\item The subspaces $\mathbf{D}^{\pm}_{2m}$ are irreducible submodules of $\mathbf{V}_{1,2m}$ under the action $\varpi_{1, 2m}$ and $\mathbf{F}_{2m-1}$ is an irreducible submodule of $\mathbf{V}_{1,2-2m}$ under $\varpi_{1, 2-2m}$. In the former case, the finite dimensional irreducible representation $\mathbf{F}_{2m-1}$ can be obtained as the subquotient as $\mathbf{V}_{1,2m}/\mathbf{D}^{-}_{2m}\oplus \mathbf{D}^{+}_{2m}\cong   \mathbf{F}_{2m-1} $. In the latter case, the discrete series $\mathbf{D}^{\pm}_{2m}$ can be realized as  
the irreducible components of the subquotient representation as 
$\mathbf{V}_{1,2-2m}/\mathbf{F}_{2m-1} \cong   \mathbf{D}^{-}_{2m}\oplus \mathbf{D}^{+}_{2m}$.
%%%%%%%
\item  The subspaces $\mathbf{D}^{\pm}_{2m+1}$ are irreducible submodule of $\mathbf{V}_{2,2m+1}$ under the action $\varpi_{2, 2m+1}$ and $\mathbf{F}_{2m}$ is an irreducible submodule of $\mathbf{V}_{2,1-2m}$ under $\varpi_{2, 1-2m}$. In the former case, the finite dimensional irreducible representation $\mathbf{F}_{2m}$ can be obtained as the subquotient as $\mathbf{V}_{2,2m+1}/\mathbf{D}^{-}_{2m+1}\oplus \mathbf{D}^{+}_{2m+1}\cong   \mathbf{F}_{2m}$, while in the latter case, the discrete series $\mathbf{D}^{\pm}_{2m+1}$ can be realized as  the irreducible components of the subquotient representation as 
$\mathbf{V}_{2,1-2m}/\mathbf{F}_{2m} \cong   \mathbf{D}^{-}_{2m+1}\oplus \mathbf{D}^{+}_{2m+1}$.
%%%%%%
\item The space $\mathbf{V}_{2,1}$ is decomposed as the irreducible sum:  
$\mathbf{V}_{2,1}=\mathbf{D}^{-}_{1}\oplus \mathbf{D}^{+}_{1}$. 
\end{enumerate}
Moreover, the spaces of irreducible submodules $\mathbf{D}^{\pm}_{m} (\subset \mathbf{V}_{j,m})$, $\mathbf{F}_{m} (\subset \mathbf{V}_{j,1-m})$ and 
the direct sum $\mathbf{D}^{+}_{m}\oplus \mathbf{D}^{-}_{m} (\subset\mathbf{V}_{j,m})$ above are the only non-trivial invariant subspaces of $\mathbf{V}_{j,m}$ for $j=1$ (resp. $j=2$) when $m$ is even (resp. odd) under the action of $\mathcal{U}(\mathfrak{sl}_2)$,
the universal enveloping algebra of $\mathfrak{sl}_2$. \qed
\end{lem}

In the lemma above, we notice that each of the following short exact sequences of $\frak{sl}_2$-modules for $m>0$ is  not split.
\begin{align}\label{non-split_SES}
& 0 \longrightarrow  \mathbf{D}^{+}_{2m}\oplus \mathbf{D}^{-}_{2m}  \longrightarrow \mathbf{V}_{1,2m}  \longrightarrow \mathbf{F}_{2m-1}  \longrightarrow 0,\\
& 0 \longrightarrow  \mathbf{D}^{+}_{2m+1}\oplus \mathbf{D}^{-}_{2m+1}  \longrightarrow \mathbf{V}_{2,2m+1}  \longrightarrow \mathbf{F}_{2m}  \longrightarrow 0.
\end{align}
%There is obviously the same structure for $V_{2,2m+1}$. 
\begin{rem}
The irreducible representations $\mathbf{D}^{-}_{1}, \, \mathbf{D}^{+}_{1}$ are called the (infinitsimal version of) limit of discrete series of $\frak{sl}_2(\R)$.
\end{rem}

%==============================================================
\section{Lie theoretical approach to the eigenspectrum}
%===============================================================

We recall first the result in \cite{WY2014JPA}. 
Let  $(\alpha, \beta, \gamma, C) \in \R^4$.  Define a second order element 
${\mathbb{K}}=\mathbb{K}(\alpha, \beta, \gamma; C) \in \mathcal{U}(\mathfrak{sl}_{2})$ and a constant $\lambda_a=\lambda_a(\alpha, \beta, \gamma)$ depending on the representation $\varpi_a$ as follows:  
\begin{align*}\label{K}
\mathbb{K}(\alpha, \beta, \gamma; C):= & \left[\frac{1}{2}H-E+\alpha \right]\left(F+\beta\right)
+ \gamma\left[H-\frac{1}{2}\right]+C,\\
\lambda_a(\alpha, \beta, \gamma):= & \beta\left(\frac12a +\alpha\right)+\gamma\left(a-\frac12\right). 
\end{align*}
By the definition of $\varpi_a$ we observe $(\partial_x:=\frac{d}{dx})$ that 
$$
%\begin{align*}
\varpi_a(\mathbb{K})
=  \Big\{x\partial_x+\frac14-\big(x^2\partial_x +\frac12(a+\frac12)x\big)+\alpha\Big\}
%& \times 
\Big\{-\partial_x+\frac12(a-\frac12)x^{-1}+\beta\Big\} + 2\gamma x\partial_x + C.
%\end{align*}
$$
Noticing $x^{-\frac12(a-\frac12)}\,x\partial_x \,x^{\frac12(a-\frac12)}= x\partial_x+ \frac12(a-\frac12)$, we have the following lemma (\cite{WY2014JPA}). 
%%%%%%%
\begin{lem}\label{K-element}
We have the following expression. 
\begin{align*}
&\frac{x^{-\frac12(a-\frac12)}\varpi_a(\mathbb{K}(\alpha, \beta, \gamma; C))x^{\frac12(a-\frac12)}}{x(x-1)}\\
= &\frac{d^2}{dx^2} +\Big\{-\beta + \frac{\frac12a+\alpha}{x} + \frac{\frac12a+2\gamma-\alpha}{x-1} \Big\}\frac{d}{dx}%\\
%&  \quad 
+  \frac{-a\beta x+\lambda_a(\alpha, \beta, \gamma)+C}{x(x-1)}. \qed
\end{align*}
\end{lem}

We construct two second order elements $\mathcal{K}$ and $\tilde{\mathcal{K}}$ $\in \mathcal{U}(\mathfrak{sl}_2)$  from ${\mathbb{K}}=\mathbb{K}(\alpha, \beta, \gamma; C)$ by suitable choices of the parameters $(\alpha, \beta, \gamma; C)$. As in the case of the quantum Rabi model in \cite{WY2014JPA}, we show that these elements provide the confluent Heun picture of the asymmetric quantum Rabi model under the image of the representations $(\varpi_{j,a},\mathbf{V}_{j,a})\,(j=1,2)$.

We now capture the confluent Heun operators $\mathcal{H}_j^{\rm{Rabi}}(\lambda)\, (j=1,2)$ through  $\mathbb{K}\in \mathcal{U}(\mathfrak{sl}_2)$. The proof of the following proposition is simple and done by the same way in \cite{WY2014JPA}. 

\begin{prop} \label{RedEigenProblem}
Let $\lambda$ be the eigenvalue of $H_{\rm{Rabi}}^\epsilon$. We have the following expressions. 
\begin{enumerate}
\item Suppose $a=-(\lambda+g^2-\epsilon)$. Define 
\begin{align*}
& {\mathcal{K}}:=  \mathbb{K}\Big(1-\frac{\lambda+g^2-\epsilon}2, 4g^2, \frac12-\frac{\lambda+g^2+\epsilon}2\,;\, \mu+4\epsilon g^2 -\epsilon^2\Big) \in  \mathcal{U}(\mathfrak{sl}_2), \\
& \Lambda_a:=  \lambda_a\Big(1-\frac{\lambda+g^2-\epsilon}2, 4g^2, \frac12-\frac{\lambda+g^2+\epsilon}2\Big).
\end{align*}
Then
\begin{equation}
x(x-1)\mathcal{H}_1^{\epsilon}(\lambda)= x^{-\frac12(a-\frac12)}(\varpi_a(\mathcal{K})-\Lambda_a)x^{\frac12(a-\frac12)}.
\end{equation}
\item Suppose $a=-(\lambda+g^2-1+\epsilon)$. Define
\begin{align*}
&\tilde{\mathcal{K}} :=  \mathbb{K}\Big(-\frac12-\frac{\lambda+g^2+\epsilon}2, 4g^2, -\frac{\lambda+g^2-\epsilon}2\,;\, \mu-4\epsilon g^2 -\epsilon^2\Big) \in  \mathcal{U}(\mathfrak{sl}_2), \\
&\tilde{\Lambda}_a:=  \lambda_a\Big(-\frac12-\frac{\lambda+g^2+\epsilon}2, 4g^2, -\frac{\lambda+g^2-\epsilon}2\Big).
\end{align*}
Then
\begin{equation}
x(x-1)\mathcal{H}_2^{\epsilon}(\lambda)= x^{-\frac12(a-\frac12)}(\varpi_a(\tilde{\mathcal{K}})-\tilde{\Lambda}_a)x^{\frac12(a-\frac12)}. \qed
\end{equation}
\end{enumerate}
\end{prop}
 
\begin{rem}
The choice of $\beta$ as $\beta=4g^2$ is necessary (and unique) if we  identify  the operator 
$\frac{x^{-\frac12(a-\frac12)}\varpi_a(\mathbb{K}(\alpha, \beta, \gamma; C))x^{\frac12(a-\frac12)}}{x(x-1)}$
with $\mathcal{H}_{j}^\epsilon(\lambda)$. 
 \end{rem}
\begin{rem}\label{intertwiner}
In Proposition \ref{RedEigenProblem}, the range of $a$ is bounded (looks almost negative). However, since there is an equivalence between $\varpi_a$ and $\varpi_{2-a}$ when $a\not\in \Z$, no such restriction applies. Actually, e.g. for the spherical case, the linear isomorphism $A:={\rm Diag}(\cdots,c_{-n},\cdots,c_0,\cdots,c_n,\cdots)$, where 
$$
c_n= c_0\prod_{k=1}^{|n|} \frac{k-\frac{a}2}{k-1+\frac{a}2} \quad (c_0\not=0),
$$
intertwines two representations $(\varpi_{1,a},\mathbf{V}_{1})$ and $(\varpi_{1,2-a},\mathbf{V}_{1})$, i.e. $A\varpi_{1,a}(X)=\varpi_{1,2-a}(X)A$ holds for any $X\in \frak{sl}_2$.  However, when $a\in 2\Z$, since  there is no such intertwiner (because there is no equivalent irreducible submodules between $\mathbf{V}_{1,2m}$ and $\mathbf{V}_{1,2-2m}$ (see Lemma \ref{Reducible})), it is neccessary to consider another approach to relating a representation $\varpi_{N}\, (N>0)$. The investigation of the relation  between discrete series representations and the spectrum of the asymmetric quantum Rabi models for a general $\epsilon\in \frac12\Z$, we will leave the details to another occasion.
\end{rem}

%\bigskip

\begin{rem}
We have a mathematical model called  the non-commutative harmonic oscillator (NcHO) introduced in  \cite{PW2001,PW2002} (see   \cite{P2010} in detail) as 
\begin{align*}
Q=Q_{\alpha,\beta}
=\begin{pmatrix}\alpha & 0 \\ 0 & \beta\end{pmatrix}\kakko{-\frac12\frac{d^2}{dx^2}+\frac12x^2}
+\begin{pmatrix}0 & -1 \\ 1 & 0\end{pmatrix}\kakko{x\frac{d}{dx}+\frac12}, \; (\alpha\beta>1, \alpha, \beta>0).
\end{align*}
which is a parity preserving (i.e. $\mathbb{Z}_2$-symmetry) self-adjoint ordinary differential operator and a generalization of the quantum harmonic oscillator having an interaction term. We observed in \cite{W2015IMRN} the close connection between NcHO and the quantum Rabi model (QRM). The QRM is obtained by a second order element $\mathcal{R} \in \mathcal{U}(\mathfrak{sl}_2)$, which is arising from NcHO through the oscillator representation. Precisely, the Heun picture of the QRM is obtained by a confluent Heun equation derived from the Heun operator defined by  $\mathcal{R}$ under a non-unitary principal series representation of $\mathfrak{sl}_2$ and confluent procedure by a finite regular singular point with $\infty$ \cite{W2015IMRN}:  
\begin{equation*}
\text{NcHO}\; {\gets}\;
{\mathcal {R}} 
\underset{\varpi_a}{\to} \text{Heun ODE} \overset{\text{non-unitary principal series}}{\longrightarrow}\underset{\text{confluent process}}{\longrightarrow}\;
\begin{matrix}
\text{Confluent Heun picture} \\ 
\text{of the QRM}
\end{matrix}
\end{equation*}
The element $\mathcal{R} \in \mathcal{U}(\mathfrak{sl}_2)$ is given as 
\begin{equation}\label{R-operator}
\mathcal{R}= 
\Big[E-F-\frac{\sqrt{a^2+1}}{a} H + \frac{\nu}{a}\Big](H-\nu)+ (\varepsilon\nu)^2,
\end{equation}
depending on non-zero parameters $(a, \nu, \varepsilon)\in \mathbb{R}^3$.
For a suitable choice of  these parameters depending on $\alpha, \beta$, $\mathcal{R} \in \mathcal{U}(\mathfrak{sl}_2)$ gives the NcHO through the oscillator representations (\cite{W2014MfI}).
It would be interesting to find a $\Z_2$-symmetry broken generalization of NcHO, which can give the asymmetric quantum Rabi model by the same confluent procedure.  Additionally, it would be desirable if there is a similar symmetry, as we will see in the subsequent sections on the asymmetric quantum Rabi model, especially in relation with the monodoromy problems for Heun ordinary equations.  Moreover, we may construct a ``$G$-function" for NcHO similar to the one defined by Braak \cite{B2011PRL}. It would be quite interesting to examine the relation between those two $G$-functions.    
\end{rem}

%=========================================================================
\section{Degeneracy of eigenstates of the asymmetric quantum Rabi model for $\epsilon \in \frac12\Z$}
%=========================================================================
In this section, we study the level crossing of the asymmetric quantum Rabi model when  $\epsilon \in \frac12\Z$ in terms of the finite dimensional irreducible representations of  $\mathfrak{sl}_2$. We first consider  the general case. 

%%%%%%%%%%%%%%%%%%%%%%%%%%%%%%%%%%%%%%%%
\subsection{$\lambda= 2m-g^2+\epsilon$ ($ \mathbf{F}_{2m+1}$: spherical $\varpi_{1,-2m}(\mathcal{K})$-eigenproblems)}
%%%%%%%%%%%%%%%%%%%%%%%%%%%%%%%%%%%%%%%%%%%%%
Let $a=-(\lambda+g^2-\epsilon)=-2m\; (m\in \Z_{>0})$. Recall the notation in Proposition \ref{RedEigenProblem}. We take 
\begin{equation*}
\begin{cases}
\alpha= 1- \frac{\lambda+g^2-\epsilon}2=1-m,\\
\beta=4g^2,\\
\gamma=\frac12- \frac{\lambda+g^2+\epsilon}2=\frac12-m-\epsilon,\\
C=\mu+4\epsilon g^2 -\epsilon^2. 
\end{cases} 
\end{equation*}

We have the following lemma. 
%-----------------------------------------
\begin{lem} Let  $\nu=\nu_{2m+1}:= \sum_{n=-m}^{m} a_n e_{1,n} \in \mathbf{F}_{2m+1}$. Then the equation $(\varpi_{1,-2m}(\mathcal{K})-\Lambda_{-2m})\nu=0$ is equivalent to the following
$$
\beta_n^+ a_{n+1} +\alpha_n^+ a_n + \gamma_n^+ a_{n-1}=0,
$$
where
\begin{equation*}
\begin{cases}
\alpha_n^+= (m-n)^2-4g^2(m-n)-\Delta^2+2\epsilon(m-n),\\
\beta_n^+=(m+n+1)(m-n-1),\\
\gamma_n^+=4g^2(m-n+1).
\end{cases}
\end{equation*}
\end{lem}
\begin{proof}
The proof can be obtained in a similar way  to that in \cite{WY2014JPA}. In fact, we observe 
\begin{align*}
\varpi_{1,-2m}(\mathcal{K})e_{1,n}
=& (m+n)(m-n)e_{1,n-1}\\
+& \Big\{(m+n)(n-1-m)+4g^2(n+1-m)+(\frac12-m-\epsilon)(2n-\frac12)+\mu+4\epsilon g^2-\epsilon^2)\Big\}e_{1,n}\\
+& 4g^2(m-n)e_{1,n+1}.
\end{align*}
It follows that 
\begin{align*}
\varpi_{1,-2m}(\mathcal{K})\nu
&= \sum_{n=-m}^{m-1} (m+n+1)(m-n-1)a_{n+1}e_{1.n}\\
&+ \sum_{n=-m}^m\Big\{-(m+n)(m-n+1)-4g^2(m-n-1)+ (\frac12-m-\epsilon)(2n-\frac12)+\mu + 4\epsilon g^2-\epsilon^2\Big\}a_n e_{1,n}\\
&+4g^2 \sum_{n=-m+1}^m(m-n+1) a_{n-1}e_{1,n}\in \mathbf{F}_{2m+1}. 
\end{align*}
Since $C=\mu+4\epsilon g^2 -\epsilon^2$, where $\mu= (\lambda+g^2)^2-4g^2(\lambda+g^2)-\Delta^2$ and 
\begin{align*}
\Lambda_a = \Lambda_{-2m}
&= \beta(\frac12a +\alpha) + \gamma(a-\frac12)\\
&= 4g^2(-2m+1)+(\frac12-m-\epsilon)(-2m-\frac12),
\end{align*}
we observe that the eigenvalue equation $(\varpi_{1,-2m}(\mathcal{K})-\Lambda_{-2m})\nu=0$ is equivalent to the stated recurrence equation. 
\end{proof}

Define now a matrix $M^{(2m,\epsilon)}_k=M_k^{(2m,\epsilon)}((2g)^2), \Delta)$ $(k=0,1,2,\cdots,2m)$ by 
%%%%%%%%%%%%%%%%%%%%%%%%%%%%%%%
\begin{align*}
M^{(2m,\epsilon)}_k=
\begin{bmatrix}
\alpha_{m}^+ & \gamma_{m}^+ &0&0& \cdots& \cdots & 0 &0 \\
\beta_{m-1}^+& \alpha_{m-1}^+ & \gamma_{m-1}^+ &0& \cdots&\cdots & \cdot & \cdot \\
0& \beta_{m-2}^+ & \alpha_{m-2}^+ & \gamma_{m-2}^+ &0& \cdots &\cdot & \cdot \\
\cdot& \cdots&\ddots& \ddots & \ddots& \cdots & 0& \cdot\\
\cdot& \cdots&\cdots& \ddots & \ddots& \gamma_{m-k+3}^+& 0& \cdot\\
\cdot& \cdots&\cdots& 0& \beta_{m-k+2}^+ & \alpha_{m-k+2}^+ & \gamma_{m-k+2}^+ & 0\\
0 & \cdots&\cdots& \cdots & 0& \beta_{m-k+1}^+ & \alpha_{m-k+1}^+ & \gamma_{m-k+1}^+ \\
0 & \cdots &\cdots& \cdots & 0 & 0 & \beta_{m-k}^+ & \alpha_{m-k}^+
\end{bmatrix}.
\end{align*}
Then the \textit{continuant} $\{\det M^{(2m,\epsilon)}_k\}_{0\leq k\leq 2m}$ (i.e., a multivariate polynomial representing the determinant of a tridiagonal matrix)  
satisfies the recurrence relation
\begin{align*}
\det M^{(2m,\epsilon)}_k 
=  \alpha_{m-k}^+ \det M^{(2m,\epsilon)}_{k-1}-\gamma_{m-k+1}^+\beta_{m-k}^+\det M^{(2m,\epsilon)}_{k-2}
\end{align*}
with initial values 
\begin{equation*}
\begin{cases}
\det M^{(2m,\epsilon)}_{0}=\alpha_m^+=-\Delta^2, \\
\det M^{(2m,\epsilon)}_{1}=\alpha_m^+\alpha_{m-1}^+=-\Delta^2(1-4g^2-\Delta^2+2\epsilon).
\end{cases}
\end{equation*}
Noticing that 
\begin{equation*}
\begin{cases}
\alpha_{m-k}^+=k^2-4g^2k-\Delta^2+2\epsilon k,\\
\beta_{m-k}^+=(2m-k+1)(k-1)\; (=(N-k+1)(k-1)) \quad \text{if} \; N=2m, \\
\gamma_{m-k+1}^+=4kg^2.
\end{cases}
\end{equation*}

%%%%%%%%%%%%%%%%%%%%%%%%%%%%%%%%%%%%%%%%
\subsection{$\lambda= 2m-g^2-1+\epsilon$ ($\mathbf{F}_{2m}$: non-spherical $\varpi_{2,1-2m}(\mathcal{K})$-eigenproblems)}
%%%%%%%%%%%%%%%%%%%%%%%%%%%%%%%%%%%%%%%%%%%%%
Let $a=-(\lambda+g^2-\epsilon)=1-2m\; (m\in \Z_{>0})$. We take 
\begin{equation*}
\begin{cases}
\alpha= 1- \frac{\lambda+g^2-\epsilon}2=\frac32-m,\\
\beta=4g^2,\\
\gamma=\frac12- \frac{\lambda+g^2+\epsilon}2=1-m-\epsilon\\
C=\mu+4\epsilon g^2 -\epsilon^2.
\end{cases} 
\end{equation*}
\begin{lem} Let  $\nu=\nu_{2m} :=\sum_{n=-m}^{m-1} b_n e_{2,n} \in \mathbf{F}_{2m}$. Then the equation $[(\varpi_{2,1-2m}(\mathcal{K})-\Lambda_{1-2m}]\nu=0$ is equivalent to the following
$$
\beta_n^- b_{n+1} +\alpha_n ^-b_n + \gamma_n^- b_{n-1}=0,
$$
where
\begin{equation*}
\begin{cases}
\alpha_n^-= (m-n-1)^2-4g^2(m-n-1-\Delta^2-2\epsilon(m-n-1),\\
\beta_n^-=(m+n+1)(m-n-2),\\
\gamma_n^-=4g^2(m-n). \qed
\end{cases}
\end{equation*}
\end{lem}

Similarly to the case of $M^{(2m,\epsilon)}_k$, we define a matrix $M^{(2m-1,\epsilon)}_k=M_k^{(2m-1,\epsilon)}((2g)^2), \Delta)$ $(k=0,1,\cdots,2m-1)$ by 
%%%%%%%%%%%%%%%%%%%%%%%%%%%%%%%
\begin{align*}
M^{(2m-1,\epsilon)}_k=
\begin{bmatrix}
\alpha_{m-1}^- & \gamma_{m-1}^- &0&0& \cdots& \cdots & 0 &0 \\
\beta_{m-2}^-& \alpha_{m-2}^- & \gamma_{m-2}^- &0& \cdots&\cdots & \cdot & \cdot \\
0& \beta_{m-3}^-& \alpha_{m-3}^- & \gamma_{m-3}^- &0& \cdots &\cdot & \cdot \\
\cdot& \cdots&\ddots& \ddots & \ddots& \cdots & 0& \cdot\\
\cdot& \cdots&\cdots& \ddots & \ddots& \gamma_{m-k+2}^- & 0& \cdot\\
\cdot& \cdots&\cdots& 0& \beta_{m-k+1}^- & \alpha_{m-k+1}^- & \gamma_{m-k+1}^-& 0\\
0 & \cdots&\cdots& \cdots & 0& \beta_{m-k}^- & \alpha_{m-k}^- & \gamma_{m-k}^-\\
0 & \cdots &\cdots& \cdots & 0 & 0 & \beta_{m-k-1}^- & \alpha_{m-k-1}^-
\end{bmatrix}.
\end{align*}
Then the continuant $\{\det M^{(2m-1,\epsilon)}_k\}_{0\leq k\leq 2m-1}$  satisfies the recurrence relation
\begin{align*}
\det M^{(2m-1,\epsilon)}_k =  \alpha_{m-k-1}^- \det M^{(2m-1,\epsilon)}_{k-1}-\gamma_{m-k}^-\beta_{m-k-1}^-\det M^{(2m-1,\epsilon)}_{k-2}
\end{align*}
with initial values 
\begin{equation*}
\begin{cases}
\det M^{(2m-1,\epsilon)}_{0}=\alpha_{m-1}^-=-\Delta^2, \\
\det M^{(2m-1,\epsilon)}_{1}=\alpha_{m-1}^-\alpha_{m-2}^-=-\Delta^2(1-4g^2-\Delta^2+2\epsilon).
\end{cases}
\end{equation*}
Noticing that 
\begin{equation*}
\begin{cases}
\alpha_{m-k-1}^-=k^2-4g^2k-\Delta^2+2\epsilon k,\\
\beta_{m-k-1}^-=(2m-k)(k-1)\; (=(N-k+1)(k-1)) \quad \text{if} \; N=2m-1, \\
\gamma_{m-k}^-=4kg^2.
\end{cases}
\end{equation*}

\bigskip

Now we define 
$$
P^{(N,\epsilon)}_k= P^{(N,\epsilon)}_k((2g)^2, \Delta^2)):= (-1)^k \det M_{k}^{(N,\epsilon)}/(-\Delta^2).
$$
Then the following lemma is obvious from the recurrence equations for $M^{(2m,\epsilon)}_k$ and 
$M^{(2m-1,\epsilon)}_k$. 
\begin{prop}
Let $x=(2g)^2$. Then $P^{(N,\epsilon)}_k(x)=P^{(N,\epsilon)}_k(x,\Delta^2)\;(k=0,1,\cdots,N)$ satisfies the following recursion formula. 
\begin{equation}
\begin{cases}
P^{(N,\epsilon)}_0=1, \quad P^{(N,\epsilon)}_1=x+\Delta^2-1-2\epsilon,\\
P^{(N,\epsilon)}_k=[kx+\Delta^2-k^2-2k\epsilon]P^{(N,\epsilon)}_{k-1} -k(k-1)(N-k+1)xP^{(N,\epsilon)}_{k-2}.
\end{cases}
\end{equation}
In particular, ${\rm deg}P^{(N,\epsilon)}_k=k$ as a polynomial in $x$. \qed
\end{prop}

We call $P^{(N,\epsilon)}_N(x)$ the {\it constraint polynomial}. It is obvious that if $x>0$ is a root of the constraint polynomial $P^{(N,\epsilon)}_N(x)$, then there is a non-zero eigenvector $\nu (\in \mathbf{F}_{N}$) of the eigenvalue equation $(\varpi_{1,a}(\mathcal{K})\nu=\Lambda_a\nu$, that is, $\nu$ gives rise  an eigenvector of the asymmetric quantum Rabi model through Proposition \ref{RedEigenProblem}.

%%%%%%%%%%%%%%%%%%%%%%%%%%%%%%%%%%%%%%%%
\subsection{$\lambda= 2m-g^2-\epsilon$ ($\mathbf{F}_{2m}$: non-spherical $\varpi_{2,1-2m}(\tilde{\mathcal{K}})$-eigenproblems)}
%%%%%%%%%%%%%%%%%%%%%%%%%%%%%%%%%%%%%%%%%%%%%
Let $a=-(\lambda+g^2-1+\epsilon)=1-2m\; (m\in \Z_{>0})$. We take 
\begin{equation*}
\begin{cases}
\alpha= -\frac12- \frac{\lambda+g^2+\epsilon}2=-\frac12-m,\\
\beta=4g^2,\\
\gamma=- \frac{\lambda+g^2-\epsilon}2=-m+\epsilon\\
C=\mu-4\epsilon g^2 -\epsilon^2.
\end{cases} 
\end{equation*}

\begin{lem} Let  $\tilde\nu=\tilde\nu_{2m}:= \sum_{n=-m}^{m-1} c_n e_{2,n} \,(c_{m-1}\not=0)\in \mathbf{F}_{2m}$. Then the equation $(\varpi_{2,1-2m}(\tilde{\mathcal{K}})-\tilde{\Lambda}_{1-2m})\tilde\nu=0$ is equivalent to the following
$$
\tilde\beta_n^- c_{n+1} +\tilde\alpha_n^- c_n + \tilde\gamma_n^- c_{n-1}=0,
$$
where
\begin{equation*}
\begin{cases}
\tilde\alpha_n^-= (m-n)^2-4g^2(m-n)-\Delta^2-2\epsilon(m-n),\\
\tilde\beta_n^-=(m+n+1)(m-n),\\
\tilde\gamma_n^-=4g^2(m-n). \qed
\end{cases}
\end{equation*}
\end{lem}

Since $c_{m+1}=c_m=0$ and $c_{m-1}\not=0$ at the equation
$$
\tilde\beta_m^- c_{m+1} +\tilde\alpha_m^- c_m + \tilde\gamma_m^- c_{m-1}=0,
$$
we find that $\tilde\gamma_m^-=0$ and may take $\tilde\alpha_{m}^-=-\Delta^2$ without loss of generality.
We define a matrix $\tilde{M}^{(2m,\epsilon)}_k=\tilde{M}_k^{(2m,\epsilon)}((2g)^2), \Delta)$ $(k=0, 1,2,\cdots,2m)$ by 
%%%%%%%%%%%%%%%%%%%%%%%%%%%%%%%
\begin{align*}
\tilde{M}^{(2m,\epsilon)}_k=
\begin{bmatrix}
\tilde\alpha_{m}^- & \tilde\gamma_{m}^-(=0) &0&0& \cdots& \cdots & 0 &0 \\
\tilde\beta_{m-1}^-& \tilde\alpha_{m-1}^- & \tilde\gamma_{m-1}^- &0& \cdots&\cdots & \cdot & \cdot \\
0& \tilde\beta_{m-2}^-& \tilde\alpha_{m-2}^- & \tilde\gamma_{m-2}^- &0& \cdots &\cdot & \cdot \\
\cdot& \cdots&\ddots& \ddots & \ddots& \cdots & 0& \cdot\\
\cdot& \cdots&\cdots& \ddots & \ddots& \tilde\gamma_{m-k+3}^- & 0& \cdot\\
\cdot& \cdots&\cdots& 0& \tilde\beta_{m-k+2}^- & \tilde\alpha_{m-k+2}^- & \tilde\gamma_{m-k+2}^-& 0\\
0 & \cdots&\cdots& \cdots & 0& \tilde\beta_{m-k+1}^- & \tilde\alpha_{m-k+1}^- & \tilde\gamma_{m-k+1}^-\\
0 & \cdots &\cdots& \cdots & 0 & 0 & \tilde\beta_{m-k}^- & \tilde\alpha_{m-k}^-
\end{bmatrix}.
\end{align*}
Then the continuant $\{\det \tilde{M}^{(2m,\epsilon)}_k\}_{0\leq k\leq 2m}$  satisfies the recurrence relation
\begin{align*}
\det \tilde{M}^{(2m,\epsilon)}_k =  \tilde\alpha_{m-k}^- \det \tilde{M}^{(2m,\epsilon)}_{k-1}-\tilde\gamma_{m-k+1}^-\tilde\beta_{m-k}^-\det \tilde{M}^{(2m,\epsilon)}_{k-2}
\end{align*}
with initial values 
\begin{equation*}
\begin{cases}
\det \tilde{M}^{(2m,\epsilon)}_{0}=\tilde\alpha_{m}^-=-\Delta^2, \\
\det \tilde{M}^{(2m,\epsilon)}_{1}=\tilde\alpha_{m}^-\tilde\alpha_{m-1}^-
=-\Delta^2(1-4g^2-\Delta^2-2\epsilon).
\end{cases}
\end{equation*}
Noticing that 
\begin{equation*}
\begin{cases}
\tilde\alpha_{m-k}^-=k^2-4g^2k-\Delta^2-2\epsilon k,\\
\tilde\beta_{m-k}^-=(2m-k+1)k \;(=(N-k+1)k) \quad \text{if} \; N=2m, \\
\tilde\gamma_{m-k+1}^-=4(k-1)g^2.
\end{cases}
\end{equation*}

%%%%%%%%%%%%%%%%%%%%%%%%%%%%%%%%%%%%%%%%
\subsection{$\lambda= 2m-g^2+1-\epsilon$ ($\mathbf{F}_{2m+1}$: spherical $\varpi_{1,-2m}(\tilde{\mathcal{K}})$-eigenproblems)}
%%%%%%%%%%%%%%%%%%%%%%%%%%%%%%%%%%%%%%%%%%%%%
Let $a=-(\lambda+g^2-1+\epsilon)=-2m\; (m\in \Z_{>0})$. We take 
\begin{equation*}
\begin{cases}
\alpha= -\frac12- \frac{\lambda+g^2+\epsilon}2=-1-m,\\
\beta=4g^2,\\
\gamma=- \frac{\lambda+g^2-\epsilon}2=-\frac12-m+\epsilon\\
C=\mu-4\epsilon g^2 -\epsilon^2.
\end{cases} 
\end{equation*}

\begin{lem} Let  $\tilde\nu= \tilde\nu_{2m+1}:=\sum_{n=-m}^{m} d_n e_{1,n} \, (d_n\not=0) \in \mathbf{F}_{2m+1}$. Then the equation $(\varpi_{1,-2m}(\tilde{\mathcal{K}})-\tilde{\Lambda}_{-2m})\tilde\nu=0$ is equivalent to the following
$$
\tilde\beta_n^+ d_{n+1} +\tilde\alpha_n^+ d_n + \tilde\gamma_n^+ d_{n-1}=0,
$$
where
\begin{equation*}
\begin{cases}
\tilde\alpha_n^+= (m-n+1)^2-4g^2(m-n+1)-\Delta^2-2\epsilon(m-n+1),\\
\tilde\beta_n^+=(m+n+1)(m-n+1),\\
\tilde\gamma_n^+=4g^2(m-n+1). \qed
\end{cases}
\end{equation*}
\end{lem}

Since $d_{m+2}=d_{m+1}=0$ and $d_{m}\not=0$ at the equation
$$
\tilde\beta_{m+1}^+ d_{m+2} +\tilde\alpha_{m+1}^+ d_{m+1} + \tilde\gamma_{m+1}^+ d_{m}=0,
$$
we find that $ \tilde\gamma_{m+1}^+=0$ and may take $\tilde\alpha_{m+1}^+=-\Delta^2$ without loss of generality. 
We define a matrix $\tilde{M}^{(2m+1,\epsilon)}_k=\tilde{M}_k^{(2m+1,\epsilon)}((2g)^2), \Delta)$ $(k=0, 1,2,\cdots,2m+1)$ by 
%%%%%%%%%%%%%%%%%%%%%%%%%%%%%%%
\begin{align*}
\tilde{M}^{(2m+1,\epsilon)}_k=
\begin{bmatrix}
\tilde\alpha_{m+1}^+ & \tilde\gamma_{m+1}^+(=0) &0&0& \cdots& \cdots & 0 &0 \\
\tilde\beta_{m}^+& \tilde\alpha_{m}^+ & \tilde\gamma_{m}^+ &0& \cdots&\cdots & \cdot & \cdot \\
0& \tilde\beta_{m-1}^+& \tilde\alpha_{m-1}^+ & \tilde\gamma_{m-1}^+ &0& \cdots &\cdot & \cdot \\
\cdot& \cdots&\ddots& \ddots & \ddots& \cdots & 0& \cdot\\
\cdot& \cdots&\cdots& \ddots & \ddots& \tilde\gamma_{m-k+4}^- & 0& \cdot\\
\cdot& \cdots&\cdots& 0& \tilde\beta_{m-k+3}^+ & \tilde\alpha_{m-k+3}^-+& \tilde\gamma_{m-k+3}^+& 0\\
0 & \cdots&\cdots& \cdots & 0& \tilde\beta_{m-k+2}^+ & \tilde\alpha_{m-k+2}^+ & \tilde\gamma_{m-k+2}^+\\
0 & \cdots &\cdots& \cdots & 0 & 0 & \tilde\beta_{m-k+1}^+ & \tilde\alpha_{m-k+1}^+
\end{bmatrix}.
\end{align*}
Then the continuant $\{\det \tilde{M}^{(2m+1,\epsilon)}_k\}_{0\leq k\leq 2m+1}$  satisfies the recurrence relation
\begin{align*}
\det \tilde{M}^{(2m+1,\epsilon)}_k =  \tilde\alpha_{m-k+1}^+ \det \tilde{M}^{(2m+1,\epsilon)}_{k-1}-\tilde\gamma_{m-k+2}^+\tilde\beta_{m-k+1}^+\det \tilde{M}^{(2m+1,\epsilon)}_{k-2}
\end{align*}
with initial values 
\begin{equation*}
\begin{cases}
\det \tilde{M}^{(2m+1,\epsilon)}_{0}=\tilde\alpha_{m+1}^+=-\Delta^2, \\
\det \tilde{M}^{(2m+1,\epsilon)}_{1}=\tilde\alpha_{m+1}^+\tilde\alpha_{m}^+
=-\Delta^2(1-4g^2-\Delta^2-2\epsilon).
\end{cases}
\end{equation*}
Noticing that 
\begin{equation*}
\begin{cases}
\tilde\alpha_{m-k+1}^+=k^2-4g^2k-\Delta^2-2\epsilon k,\\
\tilde\beta_{m-k+1}^+=(2m+2-k)k \;  (=(N-k+1)k) \quad \text{if} \; N=2m+1, \\
\tilde\gamma_{m-k+2}^+=4(k-1)g^2.
\end{cases}
\end{equation*}

\bigskip

We now define the {\it constraint polynomial} $\tilde{P}^{(N,\epsilon)}_N$ similarly to $P^{(N,\epsilon)}_N$.
$$
\tilde{P}^{(N,\epsilon)}_k= \tilde{P}^{(N,\epsilon)}_k((2g)^2, \Delta^2)):
= (-1)^k \det \tilde{M}_{k}^{(N,\epsilon)}/(-\Delta^2).
$$
Then the following lemma is obvious from the recurrence equations for $\tilde{M}^{(2m,\epsilon)}_k$ and 
$\tilde{M}^{(2m+1,\epsilon)}_k$. 
\begin{prop}
Let $x=(2g)^2$. Then $\tilde{P}^{(N,\epsilon)}_k(x,\Delta^2)\;(k=0,1,\cdots,N)$ satisfies the following recursion formula. 
\begin{equation}
\begin{cases}
\tilde{P}^{(N,\epsilon)}_0=1, \quad \tilde{P}^{(N,\epsilon)}_1=x+\Delta^2-1+2\epsilon,\\
\tilde{P}^{(N,\epsilon)}_k=[kx+\Delta^2-k^2+2k\epsilon]\tilde{P}^{(N,\epsilon)}_{k-1} -k(k-1)(N-k+1)x\tilde{P}^{(N,\epsilon)}_{k-2}.
\end{cases}
\end{equation}
In particular, ${\rm deg}\tilde{P}^{(N,\epsilon)}_k=k$ as a polynomial in $x$. \qed
\end{prop}

The following lemma is proved in \cite{LB2015JPA}  in the same way as in Ku\'s \cite{K1985JMP}.
\begin{lem} [\cite{LB2015JPA}] The constraint polynomial $P^{(N,\epsilon)}_N(x)$ (resp. $\tilde P^{(N,\epsilon)}_N(x)$) has $N-k$ positive distinct roots in the range $\sqrt{k^2+2k\epsilon}< \Delta <\sqrt{(k+1)^2+2(k+1)\epsilon}$ (with assuming $\omega=1$). \qed
\end{lem}

Moreover, the following result guaranties the crossing is ocurring at most two. 
\begin{prop}\label{Rank} 
Retain the notation and definition above (\S5.2-\S5.4). We have 
\begin{enumerate}
\item Let $x$ be a root of $P^{(N,\epsilon)}_N(x)=0$. Then the $(N+1)\times (N+1)$-matrix $M_N^{(N,\epsilon)}(x)$ has rank $N$, 
\item Let $x$ be a root of $\tilde P^{(N,\epsilon)}_N(x)=0$. Then the $(N+1)\times (N+1)$-matrix $\tilde M_N^{(N,\epsilon)}(x)$ has rank $N$. 
\end{enumerate}
\end{prop}
\begin{proof} 
Suppose $\rank M_N^{(N,\epsilon)}(x)<N$.  Then there are two linearly independent vectors $\mathbf{x}={}^t(x_0,x_1,\ldots, x_{N})$ and $\mathbf{y}={}^t(y_0,y_1,\ldots, y_{N})$ in $\ker M_N^{(N,\epsilon)}(x)$. We notice that $x_0\not=0$ (and $y_0\not=0$).  Otherwise, since $\gamma_m^+\not=0\,(N=2m+1)$ (resp. $\gamma_{m-1}^-\not=0\,(N=2m)$), we have $x_1=0$. It follows that $\mathbf{x}=0$ because the matrix is tridiagonal, whence the contradiction. Moreover, we may assume that $x_0=y_0$. Since $\mathbf{x}-\mathbf{y} \in \ker M_N^{(N,\epsilon)}(x)$, the same reasoning shows that $\mathbf{x}=\mathbf{y}$. This contradicts the assumption. Hence we observe $\rank M_N^{(N,\epsilon)}=N$.  
Although $\tilde\gamma^-_m=0 \,(N=2m)$ (resp. $\tilde\gamma^+_{m+1}=0 \,(N=2m+1)$) in the case $\tilde M_N^{(N,\epsilon)}(x)$, it is  is similarly proved because $\rank \tilde M_N^{(N,\epsilon)}(x)= \rank {}^t\tilde M_N^{(N,\epsilon)}(x)$ and 
$\tilde\beta^-_{m-1}\not=0\,(N=2m)$ (resp. $\tilde\beta^+_{m}\not=0\,(N=2m+1)$). Hence the proposition follows.\end{proof}

%====================================================
\section{Degeneracies of eigenstates for $\epsilon \in \frac12 \Z$}
%====================================================

The existence of the level crossing of eigenvalues was empirically observed in \cite{B2011PRL} when $\epsilon = \frac12$.  Recently, Li and Batchelor \cite{LB2015JPA} gave explicit evidence based on numerical computation on roots of polynomials $\tilde{P}^{(N+1,\frac12)}_{N+1}(x)$ and $P^{(N,\frac12)}_{N}(x)$, asserting  level crossings for $\epsilon=\frac12$, and claimed also their presence 
in the general case $\epsilon \in \frac12 \Z$. 

\begin{conjecture}[Li-Batchelor \cite{LB2015JPA}]\label{BLB-conjecture}
The positive roots of $P^{(N,\frac12)}_N(x)=0$ and 
$\tilde{P}^{(N+1,\frac12)}_{N+1}(x)=0$ coincide for all values of $N$ when $\epsilon=\frac12$.
\end{conjecture}

\begin{thm} The following formulas hold for any $k\,  (0\leq k \leq N)$.
\begin{equation}\label{Identity1/2}
\tilde{P}^{(N+1,\frac12)}_{k+1}(x)
=[(k+1)x+\Delta^2]P^{(N,\frac12)}_k(x) -k(k+1)(N-k)xP^{(N,\frac12)}_{k-1}(x)
\end{equation}
\end{thm}
\begin{proof} 
We prove the theorem by induction on $k$. Let $\epsilon=\frac12$. 
Since we may assume $P^{(N,\frac12)}_{-1}(x)=0$, it is true for $k=0$. It is easy to see that 
\begin{align*}
&[2x+\Delta^2]P^{(N,\frac12)}_1(x) -2(N-2)xP^{(N,\frac12)}_{0}(x)\\
=& (2x+\Delta^2)(x+\Delta^2-2) -2(N-2)x 
=(2x+\Delta^2-2)(x+\Delta^2)-2Nx \\
=&\tilde{P}^{(N+1,\frac12)}_{2}(x),
\end{align*}
whence  (\ref{Identity1/2}) holds for $k=1$. Now we assume that (\ref{Identity1/2})  is true for less than and equal to $k$. Then we observe 
\begin{align*}
&\tilde{P}^{(N+1,\frac12)}_{k+2}(x)-[(k+1)x+\Delta^2]P^{(N,\frac12)}_{k+1}(x)\\
=& \big\{(k+2)x+\Delta^2-(k+2)^2+(k+2)\big\}\tilde{P}^{(N+1,\frac12)}_{k+1}(x)
-(k+2)(k+1)(N+1-(k+2)+1)x\tilde{P}^{(N+1,\frac12)}_{k})\\
& \qquad -\{(k+2)x+\Delta^2\}P^{(N, \frac12)}_{k+1}(x)\\
=& \big\{(k+2)x+\Delta^2-(k+2)(k+1)\big\}\big\{[(k+1)x+\Delta^2]P^{(N, \frac12)}_{k}(x)-k(k+1)(N-k)xP^{(N, \frac12)}_{k-1}(x) \big\}\\
& -(k+2)(k+1)(N-k)x\big\{[kx+\Delta^2]P^{(N, \frac12)}_{k-1}(x)-k(k-1)(N-k+1)x P^{(N, \frac12)}_{k-2}(x)\big\}\\
& -[(k+2)x+\Delta^2]\big\{[(k+1)x+\Delta^2-(k+1)(k+2)]P^{(N, \frac12)}_{k}(x)-(k+1)k(N-k)xP^{(N, \frac12)}_{k-1}(x) \big\} \\
&= (k+2)(k+1)x P^{(N,\frac12)}_{k}(x) -(k+1)(k+2)(N-k)[kx+\Delta^2-k(k+1)]P^{(N,\frac12)}_{k-1}(x)\\
&-(k+2)(k+1)(N-k)xk(k-1)(N-k+1)x P^{(N, \frac12)}_{k-2}(x)\\
=& (k+2)(k+1)xP^{(N,\frac12)}_{k}(x)\\
& - (k+2)(k+1)(N-k)x\big\{[kx+\Delta^2-k(k+1)]P^{(N, \frac12)}_{k-1}(x)
-k(k-1)(N-k+1)xP^{(N, \frac12)}_{k-2}(x)\big\}\\
=& -(k+1)(k+2)(N-k)xP^{(N,\frac12)}_{k}(x).
\end{align*}
Hence the assertion of the theorem follows. 
\end{proof}

It follows immediately that 
\begin{cor}\label{CommonRoots}
\begin{equation}
\tilde{P}^{(N+1,\frac12)}_{N+1}(x)=[(N+1)x+\Delta^2]P^{(N,\frac12)}_N(x).
\end{equation}
In particular, Conjecture \ref{BLB-conjecture} is true. 
\end{cor}

Based on numerical computation of roots of the two polynomials, in \cite{LB2015JPA}, Li and Batchelor demonstrated  that  a crossing occurs when $\epsilon \in \frac12\Z$ in general. More precisely, we may state the following:
\begin{conjecture}\label{EpsilonConjecture}
For each $\epsilon =\ell/2 \in \Z_{\geq0}$  and $N \in \Z_{>0}$, 
there exists polynomial $A_N^\ell(x, \Delta^2)\in \Z[x,\Delta^2]$ satisfying the following conditions. 
\begin{enumerate}
\item $A_N^\ell(x, \Delta^2)>0$ \; \text{for all} \; $x>0$, \\
\item $\tilde{P}^{(N+\ell,\ell/2)}_{N+\ell}(x, \Delta^2)=A_N^\ell(x, \Delta^2)P^{(N,\ell/2)}_N(x,\Delta^2)$.
\end{enumerate}
\end{conjecture}

Based on this conjectural result, we may find the subsequent observation. 

Before describing the observation, we first remark the infinitesimal character of the irreducible representation. Recall the Casimir element 
\begin{equation}
\Omega := H^2+2EF+2FE \in \mathcal{U}(\frak{sl}_2).
\end{equation}
Notice that $\mathcal{Z}\mathcal{U}(\frak{sl}_2)=\C[\Omega]$, i.e. the element $\Omega$ is the generator of the center $\mathcal{Z}\mathcal{U}(\frak{sl}_2)$ of the universal enveloping algebra  $\mathcal{U}(\frak{sl}_2)$, so that $\Omega$ acts as scalar on each irreducible representation (independent from the realization of the irreducible representation). In particular, $\Omega$ commutes with both $\mathcal{K}$ and $\tilde{\mathcal{K}}$. Also, as a representation of $\mathcal{U}(\frak{sl}_2)$, $\Omega$ infinitesimally distinguishes irreducible representations by the eigenvalues in the most of the cases. 
\begin{equation}
\Omega\big|_{\mathbf{V}_{j,a}}=a(a-2), \quad \Omega\big|_{\mathbf{D}^\pm_n}=n(n-2). 
\end{equation}
In particular, we can  distinguish two irreducible  finite dimensional representations by its eigenvalues. 
$\Omega\big|_{\mathbf{F}_n}=n^2-1$.
%%%%%%%
\begin{rem} Retain the notation of \S5. Then in particular we observe
\begin{align*}
& \varpi_{1,-2m}(\Omega)\nu_{2m+1}=4m(m+1)\nu_{2m+1}, \quad
\varpi_{2,1-2m}(\Omega)\nu_{2m}=(4m^2-1)\nu_{2m},\\
&\varpi_{2,1-2m}(\Omega)\tilde\nu_{2m}=(4m^2-1)\tilde\nu_{2m}, \quad
\varpi_{1,-2m}(\Omega)\tilde\nu_{2m+1}=4m(m+1)\tilde\nu_{2m+1}.
\end{align*}
Although there is a possibility the dimension of the  eigenspace e.g. of $\varpi_{1,-2m}(\cal{K})$ with eigenvalue $\Lambda_{-2m}$ could be more than one,  by Proposition \ref{Rank},  we may also conclude the identities above. 
\end{rem}
%%%%%%%%%
\,

We now describe the structure of the eigenstate of $H_{\rm{Rabi}}^\epsilon$ for $\epsilon=0, \frac12$. 
First, we consider the case $\epsilon=0$ (the QRM case). Then we have $P_N^{(N,0)}(x)=\tilde P_N^{(N,0)}(x)$, whence if $x (>0)$ is a root of $P_N^{(N,0)}(x)$ we have the following discussion. The eigenspace of the eigenvalue $\lambda= 2m-g^2$ is given by a linear span of the vectors $\nu \in \mathbf{F}_{2m+1}$ and $\tilde\nu \in \mathbf{F}_{2m}$. Since the eigenvalue $\nu$ and $\tilde\nu$ of $\Omega$ are different, obviously $\nu(\not=0)$ and $\tilde\nu(\not=0)$ are linear independent. The case where $\lambda= 2m+1-g^2$ is similar. Hence, upon writing $W_\lambda(H_{\rm{Rabi}}^0)$ the eigenspace of $H_{\rm{Rabi}}^0(=H_{\rm{Rabi}})$ of eigenvalue $\lambda$, (under the equivalent picture stated in Theorem \ref{RedEigenProblem}) we conclude that there exist non zero elements $\nu_N, \tilde\nu_N \in \mathbf{F}_N$ such that 
\begin{align*}
W_{2m-g^2} (H_{\rm{Rabi}}^0) &\cong \C\cdot \nu_{2m+1} \oplus \C\cdot \tilde\nu_{2m} \subset \mathbf{F}_{2m+1} \oplus   \mathbf{F}_{2m},\\
W_{2m+1-g^2} (H_{\rm{Rabi}}^0) &\cong \C\cdot \nu_{2(m+1)} \oplus \C\cdot \tilde\nu_{2m+1} \subset \mathbf{F}_{2(m+1)} \oplus   \mathbf{F}_{2m+1}.
\end{align*}
These are the Judd solutions (eigenstate of exceptional degenerate eigenvalues of the quantum Rabi model $H_{\rm{Rabi}}=H_{\rm{Rabi}}^0$) obtained in Ku\'s \cite{K1985JMP}.  The reasoning we described above is the same in \cite{WY2014JPA}. 

Next we consider the case $\epsilon=\frac12$. By Corollary \ref{CommonRoots}, there exists a common root $x(>0)$ of the equations $P_N^{(N,\frac12)}(x)=0$ and $\tilde P_{N+1}^{(N+1,\frac12)}(x)=0$. Then the eigenspace $W_{2m-g^2+\frac12}(H_{\rm{Rabi}}^{\frac12})$ is spanned by $\nu \in \mathbf{F}_{2m+1}\;(P_{2m}^{(2m,\frac12)})$ and $\tilde\nu \in \mathbf{F}_{2m+1}\;(\tilde P_{2m+1}^{(2m+1,\frac12)})$. Thus, we could not see immediately if $\nu$ and $\tilde \nu$ is linearly independent. The same thing happens for $W_{2m-g^2-\frac12}(H_{\rm{Rabi}}^{\frac12})$. 
In order to clarify this problem for the case $\epsilon=\frac12$,  we make use of the following commutation relation of $\mathcal{K}$ and $\tilde{\mathcal{K}}$. The proof is done by a simple computation. 
\begin{lem} The following relation holds. 
\begin{equation}
[K, \tilde{K}]=(\epsilon+\frac32)(H+F)(F+4g^2)+
(\epsilon-\frac12)(8g^2E+HF)-2(\epsilon+\frac12)(\lambda+g^2-\frac12)F. \qed
\end{equation}
\end{lem} 

From this lemma, we have the following discussion for $\epsilon=\frac12$. Suppose that 
\begin{equation*}
\begin{cases}
\varpi_{j,a}(\mathcal{K})\nu=\Lambda_a\nu,\\
\varpi_{j,a}(\tilde{\mathcal{K}})\tilde\nu=\tilde\Lambda_a\tilde\nu.
\end{cases}
\end{equation*}
Let  $x (>0)$ be a common root of $P_{2m}^{(2m,\frac12)}(x)=0$
and $\tilde P_{2m+1}^{(2m+1,\frac12)}(x)=0$. 
Suppose that $\tilde \nu= c\nu \in \mathbf{F}_{2m+1}$ for some $c \in\C$. Then $\nu$ is a joint eigenvector of  $\varpi_{1,-2m}(\mathcal{K})$ and $\varpi_{1,-2m}(\tilde{\mathcal{K}})$. By the lemma above, we have 
$$
[K, \tilde{K}]=2(H+F)(F+4g^2)-2(\lambda+g^2-\frac12)F.
$$
Namely, $[K, \tilde{K}]$ is a linear combination of the lowerling operators $\,HF, \, F, \,F^2$ and $H$. It follows that $\nu \in \C e_{1,-m}+  \C e_{1,-m+1} \subset \mathbf{F}_{2m+1}$ since $\varpi_{1,-2m}([K, \tilde{K}])\nu=0$. But since $\varpi_{1,-2m}(H)$ preserves a vector $e_{1,-m+1}$ we find that $\nu \in \C e_{1,-m}$. Further, since $\varpi_{1,-2m}(H)$ preserves $e_{1,-m}$, we conclude that $\nu=0$. This contradicts the assumption. Hence it follows that $\dim W_{2m-g^2+\frac12}(H_{\rm{Rabi}}^{\frac12})=2$ for $m>0$. In the same way, we show that $\dim W_{2m-g^2-\frac12}(H_{\rm{Rabi}}^{\frac12})=2$.
  
 \begin{prop} If $x>0$ is a positive root of $P^{(2m,\frac12)}_{2m}(x)=0$ (resp. $P^{(2m-1,\frac12)}_{2m-1}(x)=0$),  there exist linearly independent elements $\nu_N, \tilde\nu_N \in \mathbf{F}_N$ such that 
 \begin{align*}
 W_{2m-g^2+\frac12}(H_{\rm{Rabi}}^{\frac12}) & \cong \C\cdot \nu_{2m+1} \oplus \C\cdot \tilde\nu_{2m+1},\\
 W_{2m-g^2-\frac12}(H_{\rm{Rabi}}^{\frac12}) & \cong \C\cdot \nu_{2m} \oplus \C\cdot \tilde\nu_{2m}.
 \end{align*}
 \end{prop}
 \;
 
%-------------------------------------------------------
Summarizing the discussion above, we have the following (conjectural)  reciprocity or 
$\Z_2$-sym\-me\-try behind the collection of asymmetric quantum Rabi models $\{H_{\rm{Rabi}}^\epsilon\}_{\epsilon}$. This can be interpreted as follows: the introduction of the parameter $\epsilon$ provides a resolution of singularities (i.e. removing the level crossings) by a similar idea of Riemann surfaces (the reader may find a physical idea described at p.100401-3 in \cite{B2011PRL}).

\begin{prop}
Assuming Conjecture \ref{EpsilonConjecture}, if $x>0$ is a root of polynomial $P^{(N,\epsilon)}_N(x)=0$, we have the reciprocity ($\Z_2$-symmetry) existing between two energy curves $\lambda=N\omega- g^2/\omega+\epsilon$ and $\lambda=N'\omega- g^2/\omega-\epsilon$ for $\epsilon \in \frac12 \Z$ described in 
Table \ref{Symmetry1} and Table \ref{Symmetry2}: 
{
\begin{table}%
\caption{$H_{\rm{Rabi}}^\ell$ for $\epsilon=\ell\geq 0, \; \ell\in \Z,\, m\in \Z_{>0}$}
\label{Symmetry1}
\begin{center}
\begin{tabular}{@{\vrule width1.0pt~}l|cc|%
cc|cc@{~\vrule width 1.0pt}}
\noalign{\hrule height2.0pt}
$\lambda+g^2$ &\multicolumn{2}{p{4em}|}
{$2m+\ell$} & 
\multicolumn{2}{p{4em}@{~\vrule width 1.0pt}}%
{$2m+\ell+1$}\\
\hline
$W_\lambda=\C\nu\oplus \C\tilde\nu$ & $\mathbf{F}_{2m+1}$ & $\oplus \;\;\qquad \mathbf{F}_{2(m+\ell)}$  
& $\mathbf{F}_{2(m+1)}$ & $\oplus \qquad \mathbf{F}_{2(m+\ell)+1}$\\
Common roots & ($P^{(2m, \ell)}_{2m}$, & $\tilde P^{(2(m+\ell), \ell)}_{2(m+\ell)}$) 
& ($P^{(2m+1, \ell)}_{2m+1}$, & $\tilde P^{(2(m+\ell)+1, \ell)}_{2(m+\ell)+1}$) \\
Parity & spherical & $\oplus \;\;$ non-spherical & non-spherical & $\oplus \quad$ spherical\\
\noalign{\hrule height1.0pt}
\end{tabular}
\end{center}
\end{table}

%------------------------------------------------
\begin{table}%
\caption{$H_{\rm{Rabi}}^{\ell\pm\frac12}$ for $\epsilon=\ell\pm\frac12>0, \; \ell\in \Z,\, m\in \Z_{>0}$}
\label{Symmetry2}
\begin{center}
\begin{tabular}{@{\vrule width1.0pt~}l|cc|%
cc|cc@{~\vrule width 1.0pt}}
\noalign{\hrule height2.0pt}
$\lambda+g^2$ &\multicolumn{2}{p{4em}|}
{$2m+\ell+\frac12$} & 
\multicolumn{2}{p{4em}@{~\vrule width 1.0pt}}%
{$2m+\ell-\frac12$}\\
\hline
$W_\lambda =\C\nu\oplus \C\tilde\nu$ & $\mathbf{F}_{2m+1}$ & $\oplus \qquad \mathbf{F}_{2(m+\ell)+1}$  
& $\mathbf{F}_{2m}$ & $ \oplus \;\qquad \mathbf{F}_{2(m+\ell)}$\\
Common roots & ($P^{(2m, \ell+\frac12)}_{2m}$, & $\tilde P^{(2(m+\ell)+1, \ell+\frac12)}_{2(m+\ell)+1}$) 
& ($P^{(2m-1, \ell+\frac12)}_{2m-1}$, & $\tilde P^{(2(m+\ell), \ell+\frac12)}_{2(m+\ell)}$) \\
Parity & spherical & $\oplus \qquad$ spherical & non-spherical & $\oplus \quad$ non-spherical\\
\noalign{\hrule height1.0pt}
\end{tabular}
\end{center}
\end{table}
}
\qed
\end{prop}

 \begin{rem}
 Since $\dim W_\lambda\leq 2$, all the level crossings can appear only in the cases described in Table \ref{Symmetry1} and Table \ref{Symmetry2}. This is consistent with the fact, e.g. that $\mathbf{V}_{1,2-2m}$ can not have the direct sum decomposition such as $\mathbf{V}_{1,2-2m} \cong \mathbf{D}^{+}_{2m}\oplus \mathbf{D}^{-}_{2m} \oplus \mathbf{F}_{2m-1} $  by (\ref{non-split_SES}) (similarly, this holds for $\mathbf{V}_{2,1-2m}, \mathbf{V}_{1, 2m}$ and $\mathbf{V}_{2, 2m+1}\;(m>0)$).
 \end{rem}
 
\,

\begin{prob}  
The finite dimensional irreducible representation $\mathbf{F}_N$ can be lifted to the representation of the compact Lie group $U(2)$, the unitary group of degree $2$. (Note that the complexified Lie algebras $\frak{u}(2)_{\C}$ and $\frak{sl}_2(\R)_{\C}$ are identified with $\frak{sl}_2(\C)$ so that the finite dimensional irreducible representations are all the same.)Suppose $\epsilon \in \frac12\Z$. 
The question is whether there exists any discrete (= finite) subgroup $\Gamma=\Gamma_{\epsilon,\Delta^2}$  of  $U(2)$ independent from $N$ such that the eigenspace $W_\lambda$ of level crossings (or degenerate) can be realized as subspaces in $L^2(U(2)/\Gamma)$?  In other words, is there any finite subgroup $\Gamma=\Gamma_{\epsilon,\Delta^2}$ which produces 
the following equivalence? 
$$
\varpi_{N-g^2+\epsilon}\big|_{\Gamma}=I  \quad \Leftrightarrow \quad P^{(N, \epsilon)}_N((2g)^2)=0,  \quad (\forall N\in \Z_{>0}).
$$
Here we make an abuse of notation for the representation $\varpi_{N-g^2+\epsilon}$  of Lie algebra $\frak{u}(2)$ also as a representation of Lie group $U(2)$. 
\end{prob}

%====================================================
\section{Non-degenerate exceptional eigenvalues}
%====================================================

Recall the spectral situation of the quantum Rabi model. Denote the spectrum (set of all eigenvalues) of the quantum Rabi model Hamiltonian $H_{\rm{Rabi}}$ by $\Sigma_{\rm{Rabi}}$. Then we have (\cite{MPS2013, WY2014JPA, B2013MfI})
\begin{align*}
\Sigma_{\rm{Rabi}} & = \sigma_{\rm{reg.}} \coprod \sigma_{\rm{deg.excep.}}
\coprod \sigma_{\rm{non-deg. excep.}},
\end{align*}
where  $\sigma_{\rm{reg.}}$ denotes the set of all elements of the regular spectrum, $\sigma_{\rm{deg.excep.}}$ the degenerate exceptional  spectrum and $\sigma_{\rm{non-deg.excep.}}$ the non-degenerate exceptional spectrum. Notice that the non-degenerate spectrum is consisting of $\sigma_{\rm{reg.}} $ and $\sigma_{\rm{non-deg.excep.}}$.
In this section, we consider the non-degenerate spectrum of the (asymmetric) quantum Rabi model. 

If $\phi$ is the eigenfunction of the Hamiltonian $H_{\rm{Rabi}}^\epsilon$, then the corresponding vector 
$\nu_{\phi}$ in Proposition \ref{RedEigenProblem} is an eigenvector of the operator $\varpi_{a}(\mathcal{K})$ (or $\varpi_{a}(\tilde{\mathcal{K}})$). Since $\Omega\in \mathcal{Z}(\mathcal{U}(\frak{sl}_2))$,
we find that $\nu_\phi$ is also an eigenvector of $\varpi_a(\Omega)$ provided the corresponding eigenstate of  $\varpi_{a}(\mathcal{K})$ (or $\varpi_{a}(\tilde{\mathcal{K}})$) is known to be non-degenerate (a priori). 
From this observation, we note that we have either $\varpi_{j,n}(\Omega)|_{\mathbf{D}^\pm_n}= n(n-2){I}$ or $\varpi_a(\Omega) =a(a-2){I} \,(a\not\in\Z)$, ${I}$ being the identity operator of appropriated suggesting space. 
Relating this observation, we have the following remark. 
%-------------------------
\begin{lem}\label{IrreducibleRepre}
Let $\lambda$ be a non-degenerate eigenvalue of the Hamiltonian $H_{\rm{Rabi}}^\epsilon$. Then the corresponding eigenvector $\nu$ in $\mathbf{V}_{j, a}$ for $\varpi_{a}(\mathcal{K})$ with eigenvalue $\Lambda_a$ (or $\varpi_{a}(\tilde{\mathcal{K}})$ with eigenvalue $\tilde{\Lambda}_a$) via Proposition \ref{RedEigenProblem} is captured in the space of an irreducible representation, that is, either of  $\mathbf{D}^{\pm}_{N}$ and $\mathbf{V}_{j,a}\,(a\not\in \Z)$. 
\end{lem}
%------------------------
\begin{proof}
Obviously, it is enough to consider the case where $a \in \Z$.  Let us consider the representation $(\varpi_{1, 2m},\,\mathbf{V}_{1,2m})$. The case  $(\varpi_{2, 2m-1},\,\mathbf{V}_{2,2m-1})\,(m>0)$ is similar. We may assume $a=m>0$ by the discussion in \S6. Suppose $\nu\in V_{1,2m}$ satisfies $\varpi_{1,2m}(\mathcal{K})\nu=\Lambda_{2m}\nu$. Write $\nu=\nu_m^+ + \nu_0 + \nu_m^-$, where $\nu^\pm_m\in \mathbf{D}_{2m}^{\pm}$ and $\nu_0\in \mathbf{F}_{2m-1}$. Since $\varpi_{1,2m}(\mathcal{K})\nu_m^\pm \in \mathbf{D}_{2m}^\pm$ we have 
$\varpi_{1,2m}(\mathcal{K})\nu_m^\pm= \Lambda_{2m}\nu_m^\pm$. It follows that $\varpi_{1,2m}(\mathcal{K})\nu_0=\Lambda_{2m}\nu_0$. Notice that the space $\mathbf{F}_{2m-1}$ is not invariant under the representation $\varpi_{1,2m}$ of $\frak{sl}_2$ (see (\ref{non-split_SES})). Therefore, to prove the assertion, it is sufficient to show that the vector $\nu_0=\sum_{n=1-m}^{m-1}t_n e_{1,n}\,\in \mathbf{F}_{2m-1}$ cannot be the eigenvector of this eigenvalue problem.  
Note that $\mathcal{K}$ is a linear combination  of $E,F, H, EF$ with non-zero coefficients and constant. By the existence of $E$ and $F$, it is immediate to see that $t_{1-m}=t_{m-1}=0$. Then, similarly to the computation we have done in \S5, we have $t_{2-m}=t_{m-2}=0$ and so on, and conclude that $\nu_0=0$.
Hence the eigenvector $\nu=\nu_m^+ + \nu_m^- \in \mathbf{D}^{+}_{2m}\oplus \mathbf{D}^{-}_{2m}$. Since  $\mathbf{D}^{+}_{2m}\oplus \mathbf{D}^{-}_{2m}$ is the direct sum, obviously either of 
$\nu_m^\pm$ is zero, otherwise the eigenspace becomes two dimensional. This proves the lemma.
\end{proof}

We will restrict ourselves to the representation theoretical discussion to the case $\epsilon=0$, i.e. the quantum Rabi model. We leave the general case to future investigation. 
 
We discuss the non-degenerate exceptional spectrum of the quantum Rabi model. The level crossing is captured by the finite dimensional irreducible modules of $\mathfrak{sl}_2(\mathbb{R})$, while there is another possibility of the existence of non-degenerate eigenvalues of the form $\lambda=N-g^2\, (N\in \mathbb{Z})$ associated with the discrete series representations $\mathbf{D}^\pm_N$ of $\mathfrak{sl}_2(\mathbb{R})$ in the case of pairs $(g, \Delta)$ satisfying $P^{(N, \epsilon)}_{N}(4g^2, \Delta^2)\not=0$.
The observation given here agrees with the investigation and numerical observation
given in the recent paper \cite{MPS2013} (see also \cite{B2013JPB}), especially for a large $\Delta$ in view of the results in the preceding subsection.

Let us recall the discussion  in \cite{B2013MfI} for the non-degenerate exceptional eigenvalues. Then we may find that the non-degenerate exceptional eigenstate is obtained in the discrete series representation
$\mathbf{D}^{+}_N$. Recall the notation and convention in Section \ref{CHP}.　We put $f_\pm(z)=e^{-gz}\phi_{1,\pm}(x)$ for $x=(z+g)/(2g)$. (Note that  $\phi_1=\phi_{1,+}$) Put $a=-(\lambda+g^2)$. The  equations (\ref{SDE1}) for $f_\pm$ turn to be the following:
\begin{equation}
\begin{cases}\label{SDE2}
x\frac{d}{dx} \phi_{1,+}(x) = -a \phi_{1,+}(x) -\Delta \phi_{1,-}(x),\\
(x-1)\frac{d}{dx} \phi_{1,-}(x)= (-a-4g^2+4g^2 x)  \phi_{1,-}(x) -  \Delta \phi_{1,+}(x) .
\end{cases}
\end{equation}
Let $a=-N\in \Z_{>0}$. The exponents of the function $\phi_{1,-}(x)$ (resp. $\phi_{1,+}(x)$) at $x=0$ is $\{0, N+1\}$ (resp. $\{0, N\}$) (Lemma \ref{Exponents}). Since the difference of the two exponents at $x=0$ (and $x=1$) is a positive integer $N$, the local analytic Frobenius solutions around $x=0$ will develop a logarithmic branch-cut at $x=1$ in general. However, since the largest exponent of $\phi_{1,-}(x)$ at $x=0$ is $N+1$, there exists always a solution for $\phi_{1,-}(x)$ analytic at $x=0$ of the form
\begin{equation}
\phi_{1,-}(x) = \sum_{n=N+1}^\infty K_n(N; g, \Delta) x^n.
\end{equation}
Put $K_n=K_n(N; g, \Delta)$. Keeping $a=-N$ and integrating the first equation of (\ref{SDE2}), we observe that
\begin{equation}
\phi_{1,+}(x) = c x^N - \Delta \sum_{n=N+1}^\infty \frac{K_n}{n-N} x^n
\end{equation}
for some constant $c$. Plugging these expressions for $\phi_{1,\pm}(x)$ into (\ref{SDE2}), we observe 
the recurrence equation 
\begin{equation}\label{RecRel_K}
(n+1)K_{n+1}=\big(4g^2+n-N+ \frac{\Delta^2}{N-n}\big)K_n -4g^2 K_{n-1} \quad (n>N).
\end{equation}
with the initial conditions
$$
(N+1)K_{N+1} =c\Delta, \quad K_N=0.
$$
Hence, defining $K_{N+1}=1$, we have 
\begin{equation}
\phi_{1}(x) = \phi_{1,+}(x) = \frac{N+1}\Delta x^N - \Delta \sum_{n=N+1}^\infty \frac{K_n}{n-N} x^n. 
\end{equation}
In other words, by Proposition \ref{RedEigenProblem}, when $N=2m(=-a)$ we have 
$$
\nu_{\phi_1}:= x^{\frac12(a-\frac12)} \phi_{1}(x)= x^{-m-\frac14}\phi_{1}(x)
= \frac{2m+1}\Delta e_{1,m} - \Delta \sum_{n=2m+1}^\infty \frac{K_n}{n-2m} e_{1,n-m} \in \mathbf{D}^+_{2m}.
$$
If $N=2m-1(=-a)$, similarly we have
$$
\tilde\nu_{\phi_1}:= x^{\frac12(a-\frac12)} \phi_{1}(x)= x^{-m+\frac14}\phi_{1}(x)
= \frac{2m}\Delta e_{2,m-1} - \Delta \sum_{n=2m}^\infty \frac{K_n}{n-2m+1} e_{1,n-m} \in \mathbf{D}^-_{2m-1}.
$$

%Thanks to Lemma \ref{IrreducibleRepre},
%combining the discussion in \cite{B2013MfI}, 
We have the following
\begin{prop} Let $N\in \Z_{>0}$. Let  $G_+^{(N)}(g,\Delta)$ be the $G$-function \cite{B2013MfI} given by 
\begin{equation}
G_+^{(N)}(g,\Delta)=-\frac{2(N+1)}\Delta +\sum_{n=N+1}^\infty K_n(N; g\Delta)\Big(1+\frac{\Delta}{n-N}\Big)g^{n-N-1}.
\end{equation}
If $g$ is a zero of $G_+^{(N)}(g,\Delta)$ (whence $P^{(N, 0)}_{N}(4g^2, \Delta^2)\not=0$), we have a non-degenerate exceptional eigenvalue $\lambda= N-g^2$ such that 
$\nu_{\phi_1}$ is contained in (the space of ) $\mathbf{D}^+_N$. The same statement holds for 
$G_-^{(N)}(g,\Delta)$.  
\qed
\end{prop}

\begin{rem}
Since $G_-^{(N)}(g,\Delta)=G_+^{(N)}(g,-\Delta)$, there is no chance that they have a common root $g$. 
\end{rem}

\begin{rem}
Although there is no irreducible submodule other than $\mathbf{F}_{2m+1}$ in $\mathbf{V}_{1,-2m}$, it follows from the discussion above that $\nu_{\phi_1}$ is a non-zero eigenvector of the equation 
$$
\varpi_{1,-2m}(\mathcal{K})\nu_{\phi_1}= \Lambda_{-2m}\nu_{\phi_1}. %\mod(\mathbf{F}_{2m-1}),.
$$
We may also confirm this equation directly by the recursion relation (\ref{RecRel_K}) for $K_n$.  
%Here $(\varpi_{1,-2m},\,V_{1,-2m}/\mathbf{F}_{2m+1})$ is the subquotient representation of $\mathcal{U}(\frak{sl}_2)$ on $V_{1,-2m}/\mathbf{F}_{2m+1}\cong \mathbf{D}_{2(m+1)}^+\oplus \mathbf{D}_{2(m+1)}^-$. 
The case $N=2m-1$ is similarly observed. 
%We will discuss more deeply for general $\epsilon \in \frac12\Z$. The detail of the problems including general $\epsilon$ cases will be left in the future study (cf. Remark \ref{intertwiner}). 
\end{rem}

\begin{rem}
None of the limit of discrete series $\mathbf{D}^{\pm}_1$ contributes any level crossing. 
\end{rem}

\begin{rem}
As we discussed in \cite{WY2014JPA}, the regular eigenvalues $\lambda$ are considered to be captured in the non-unitary principal series $\mathbf{V}_{j,a}\, (j=1,2,\, a\not \in \Z)$ (Lemma \ref{IrreducibleRepre}). The "parity" $j$ is determined by the fact that $x=\lambda+g^2$ is either a zero of $G_+(x)$  or $G_-(x)$ \cite{B2011PRL, B2013AP}. 
\end{rem}

%%%%%%%%%%%%%%%%%%%%%%%%%%%%%%%%%%%%
\section*{Acknowledgements}
%%%%%%%%%%%%%%%%%%%%%%%%%%%%%%%%%%%%%
The author wishes to thank Daniel Braak for his helpful comments and suggestions, particularly  from a physical perspective. The present paper would have never appeared without his encouragement. This work is partially supported by Grand-in-Aid for Scientific Research (C) No. 16K05063 of JSPS
and by CREST, JST, Japan.

%%%%%%%%%%%%%%%%%%%%%%%%%%%%%%%%%%%%

\begin{flushleft}
\bigskip

Masato Wakayama \par
Institute of Mathematics for Industry,\par
Kyushu University \par
744 Motooka, Nishi-ku, Fukuoka 819-0395, Japan \par
\texttt{wakayama@imi.kyushu-u.ac.jp}

\end{flushleft}

\end{document}